\documentclass[a4paper,UKenglish]{lipics}
 
\usepackage{microtype}
\usepackage{amsmath}
\usepackage{amsthm}
\usepackage{amsfonts}
\usepackage{enumerate}
\usepackage{bm}
\usepackage{graphicx}
\usepackage{gensymb}
\usepackage{xifthen}

\renewcommand{\int}[1]{$(\mathit{int}\cup{#1})$}

\newcommand{\bb}[2][]{
  \ifthenelse{\isempty{#1}}
	{\mathbf{#2}}
	{\mathbf{#2^{\mathnormal{#1}}}}
}

\newtheorem{claim}[theorem]{Claim}

\title{$1$-String  $B_2$-VPG Representation of Planar Graphs\footnote{Research supported by NSERC. The second author was supported by the Vanier CGS. A preliminary version appeared at the {\em Symposium on Computational Geometry 2015}.}}
\titlerunning{$\bm{1}$-String  $\bm{B_2}$-VPG Representation of Planar Graphs} 

\author[1]{Therese Biedl}
\author[1]{Martin Derka}
\affil[1]{David R.~Cheriton School of Computer Science, University of Waterloo\\
  200 University Ave W, Waterloo, ON N2L 3G1, Canada\\
  \texttt{\{biedl,mderka\}@uwaterloo.ca}}
\authorrunning{T. Biedl and M. Derka} 

\Copyright{Therese Biedl and Martin Derka}


\serieslogo{}
\volumeinfo
  {Billy Editor and Bill Editors}
  {2}
  {Conference title on which this volume is based on}
  {1}
  {1}
  {1}
\EventShortName{}
\DOI{10.4230/LIPIcs.xxx.yyy.p}

\begin{document}

\maketitle

\newcommand{\rev}[0]{\ensuremath{^{\mbox{\scriptsize{\,rev}}}}}

\begin{abstract}
In this paper, we prove that every planar graph has a 1-string $B_2$-VPG representation---a
string representation using paths in a rectangular grid that contain at most two bends. 
Furthermore, two paths representing vertices $u,v$ intersect precisely once whenever there
is an edge between $u$ and $v$. We also show that only a subset of the possible curve shapes is necessary to represent $4$-connected planar graphs.
\end{abstract}

\section{Preliminaries}

One way of representing graphs is to assign to every vertex a curve so
that two curves cross if and only if there is an edge between the respective vertices. 
Here, two curves $\bb{u},\bb{v}$ \emph{cross} means that
they share a point $s$ internal to both of them and 
the boundary of a sufficiently small closed disk around $s$ 
is crossed by $\bb{u},\bb{v},\bb{u},\bb{v}$ (in this order).
Such a representation of graphs using crossing curves is referred to as a \emph{string representation},
and graphs that can be represented in this way are called \emph{string graphs}.

In 1976, Ehrlich, Even and Tarjan showed that every planar graph has a string representation~\cite{cit:tarjan}.
It is only natural to ask if this result holds if one is restricted to using
only some ``nice'' types of curves. In 1984, Scheinerman conjectured that all planar graphs can
be represented as intersection graphs of line segments~\cite{cit:scheinerman}.
This was proved first for bipartite planar graphs~\cite{cit:pach, cit:arroyo} with the strengthening
that every segment is vertical or horizontal. The result was extended to 
triangle-free planar graphs, which can be represented by line segments with at most three distinct slopes~\cite{cit:castro}.

Since Scheinerman's conjecture seemed difficult to prove for all planar
graphs, interest arose in possible relaxations.
Note that any two line segments can intersect at most once.
Define \textsc{1-String} to be the class of graphs that are intersection graphs
of curves (of arbitrary shape) that intersect at most once. We also say that
graphs in this class have a \emph{$1$-string representation}.
The original construction of string representations for planar graphs 
given in~\cite{cit:tarjan} requires curves to cross multiple times. 
In 2007, Chalopin, Gon\c{c}alves and Ochem showed that every
planar graph is in \textsc{1-String}~\cite{cit:chalopin-gonclaves-ochem, cit:chalopin-string}.  With respect to Scheinerman's
conjecture, while the argument of~\cite{cit:chalopin-gonclaves-ochem, cit:chalopin-string} shows that the prescribed number
of intersections can be achieved, it provides no idea on the complexity of curves that is required.   

Another way of restricting curves in string representations is to require them
to be \emph{orthogonal}, i.e., to be paths in a grid.  Call a graph a
{\em VPG-graph} (as in ``Vertex-intersection graph of Paths in a Grid'')
if it has a string representation with orthogonal curves.
It is easy to see that all planar graphs are VPG-graphs (e.g.~by generalizing
the construction of Ehrlich, Even and Tarjan).  For bipartite planar graphs,
curves can even be required to have no bends \cite{cit:pach, cit:arroyo}.
For arbitrary planar graphs, bends are required in orthogonal curves.
Recently, Chaplick and Ueckerdt showed that two bends per curve always suffice
\cite{cit:chaplick}.  Let {\em $B_2$-VPG} be the graphs that have
a string representation where curves are
orthogonal and have at most two bends; the result in
\cite{cit:chaplick} then states that planar graphs are in $B_2$-VPG.
Unfortunately, in Chaplick and Ueckerdt's construction, curves may cross 
each other twice, and so it
does not prove that planar graphs are in \textsc{1-String}.

The conjecture of Scheinerman remained open until 2009 when it was proved true by Chalopin and Gon\c{c}alves~\cite{cit:chalopin-seg}.

\medskip\noindent{\bf Our Results: }
In this paper, we show that every planar graph has a 
string representation that simultaneously satisfies the requirements for
\textsc{1-String} (any two curves cross at most once) and the requirements
for $B_2$-VPG (any curve is orthogonal and has at most two bends).
Our result hence re-proves, in one construction, the results
by Chalopin et al.~\cite{cit:chalopin-gonclaves-ochem, cit:chalopin-string} and the result by
Chaplick and Ueckerdt \cite{cit:chaplick}. 

\begin{theorem}
\label{thm:main-claim}
Every planar graph has a $1$-string $B_2$-VPG rep\-re\-sen\-ta\-tion. 
\end{theorem}

In addition to Theorem~\ref{thm:main-claim}, we show that for $4$-connected planar graphs, only a subset of orthogonal curves with $2$ bends is needed:

\begin{theorem}
\label{thm:cz}
Every $4$-connected planar graph has a $1$-string $B_2$-VPG rep\-re\-sen\-ta\-tion where all curves have a shape of C or Z (including their horizontal mirror images). 
\end{theorem}

Our approach is inspired by the construction of 1-string representations from 
2007~\cite{cit:chalopin-gonclaves-ochem, cit:chalopin-string}. 
The authors proved the result in two steps. First,
they showed that triangulations without separating triangles 
admit 1-string representations. By induction on the number of 
separating triangles, they then showed that a 1-string representation
exists for any planar triangulation, and consequently for any 
planar graph. 

In order to show that triangulations without separating triangles
have 1-string representations, Chalopin et al.~\cite{cit:chalopin-string} used
a method inspired by Whitney's proof that 4-connected planar graphs
are Hamiltonian~\cite{cit:whitney}. Asano, Saito and Kikuchi later improved
Whitney's technique and simplified his proof~\cite{cit:ham-cycle}. 
Our paper uses the same approach as~\cite{cit:chalopin-string}, but borrows ideas from~\cite{cit:ham-cycle}
and develops them further to reduce the number of cases.

\section{Definitions and Basic Results}

Let us begin with a formal definition of a \emph{$1$-string $B_2$-VPG representation}.

\begin{definition}[$1$-string $B_2$-VPG representation]
A graph $G$ has a \emph{1-string $B_2$-VPG representation} if every vertex $v$ of $G$ can be
represented by a curve $\mathbf{v}$ such that: 
\begin{enumerate}
    \item Curve $\mathbf{v}$ is \emph{orthogonal}, i.e., it consists of horizontal and vertical segments.
    \item Curve $\mathbf{v}$ has at most two bends.
    \item Curves $\mathbf{u}$ and $\mathbf{v}$ intersect at most once, and $\bb{u}$ intersects $\bb{v}$ if and
	only if $(u,v)$ is an edge of $G$.
\end{enumerate}
\end{definition}

We always use $\bb{v}$ to denote the curve of vertex $v$, and write $\bb[R]{v}$ if the
representation $R$ is not clear from the context. We also often omit ``1-string $B_2$-VPG''
since we do not consider any other representations.

Our technique for constructing representations of a graph uses an intermediate step
referred to as a ``partial $1$-string $B_2$-VPG representation
of a W-triangulation that satisfies the chord condition with respect to three chosen corners.'' 
We define these terms, and related graph terms, first.  

A {\em planar graph} is a graph that can be embedded in the plane, i.e., it can be
drawn so that no edges intersect except at common endpoints.  All graphs in this
paper are planar.  We assume throughout the paper that one 
\emph{combinatorial embedding} of the 
graph has been fixed by specifying the clockwise (CW) cyclic order of incident edges around
each vertex.  Subgraphs inherit this embedding, i.e., they use the induced clockwise orders.
A {\em facial region} is a connected region of $\mathbb{R}_2-\Gamma$ where $\Gamma$ is a planar
drawing of $G$ that conforms with the combinatorial embedding.  
The circuit bounding this region
can be read from the combinatorial embedding of $G$ and is referred to as a {\em facial circuit}. We sometimes refer to both facial circuit and facial region as a \emph{face} when the precise meaning is clear from the context. The
{\em outer-face} is the one that corresponds to the unbounded region; all others are
called {\em interior faces}.  The outer-face cannot be read from the embedding; we
assume throughout this paper that the outer-face of $G$ has been specified.  Subgraphs inherit
the outer-face by using as outer-face the one whose facial region contains the facial
region of the outer-face of $G$.    An edge of $G$ is called {\em interior} if it does
not belong to the outer-face.

A \emph{triangulated disk} is a planar graph $G$ for which the outer-face is a simple
cycle and every interior face is a triangle. 
A {\em separating triangle} is a cycle $C$ of length $3$ such that $G$ has
vertices both inside and outside the region bounded by $C$ (with respect to the
fixed embedding and outer-face of $G$). 
Following the notation of
\cite{cit:chalopin-string}, a \emph{W-triangulation} is a triangulated disk
that does not contain a separating triangle.
A {\em chord} of a triangulated disk is an
interior edge for which both endpoints are on the outer-face.

Let $X, Y$ be two vertices on the outer-face of a connected planar graph so that neither of them is a cut vertex. Define $P_{XY}$ to be the counter-clockwise (CCW) path on the outer-face from $X$ to $Y$ ($X$ and $Y$ inclusive).
We often study triangulated disks with
three specified distinct vertices $A,B,C$ called the {\em corners}. 
$A,B,C$ must appear on the outer-face in CCW order. 
We denote $P_{AB} = (a_1, a_2, \ldots, a_r)$, $P_{BC} = (b_1, b_2, \ldots, b_s)$ 
and $P_{CA} = (c_1,c_2,\ldots,c_t)$, where $c_t = a_1 = A$, $a_r = b_1 = B$ 
and $b_s = c_1 = C$.

\begin{definition}[Chord condition]
\label{def:chord-condition}
A W-triangulation $G$ satisfies the \emph{chord condition} with respect
to the corners $A,B,C$ if $G$ has no chord within $P_{AB}, P_{BC}$ or $P_{CA}$,
i.e., no interior edge of $G$ has 
both ends on $P_{AB}$, or both ends on $P_{BC}$, or
both ends on $P_{CA}$.\footnote{For readers familiar with \cite{cit:chalopin-string}
or \cite{cit:ham-cycle}:
A W-triangulation that satisfies the chord condition with respect
to corners $A,B,C$ is called a \emph{W-triangulation 
with 3-boundary $P_{AB},P_{BC},P_{CA}$} 
in \cite{cit:chalopin-string},
and the chord condition is 
the same as \emph{Condition (W2b)} in~\cite{cit:ham-cycle}.}
\end{definition}

\begin{definition}[Partial $1$-string $B_2$-VPG representation]
Let $G$ be a connected planar graph and $E' \subseteq E(G)$ be a set of edges. 
An {\em $(E')$-1-string $B_2$-VPG representation} of $G$ is a 1-string $B_2$-VPG representation of 
the subgraph $(V(G),E')$, i.e.,
curves $\bb{u},\bb{v}$ cross if and only if $(u,v)$ is an edge in $E'$.
If $E'$ consists of all interior edges of $G$ as well as some set of edges $F$ on the
outer-face, then we write \emph{\int{F} representation} instead. 
\end{definition}

In our constructions, we use \int{F} representations with $F=\emptyset$ or $F=\{e\}$, where $e$
is an outer-face edge incident to corner $C$ of a W-triangulation.   
Edge $e$ is called the \emph{special edge}, and we sometimes write
\int{e} representation, rather than \int{\{e\}} representation.

\subsection{$2$-Sided, $3$-Sided and Reverse $3$-Sided Layouts}

To create representations where vertex-curves have few bends, we need
to impose geometric restrictions on representations of subgraphs.  
Unfortunately, no one type of layout seems sufficient for all cases,
and we will hence have three different layout types 
illustrated in Figure~\ref{fig:layoutTypes}.

\begin{definition}[$2$-sided layout]
Let $G$ be a connected planar graph and $A,B$ be two distinct outer-face vertices neither of which is a cut vertex in $G$. Furthermore, let $G$ be such that all cut vertices separate it into at most two connected components.  
An \int{F} 
$B_2$-VPG representation of $G$ (for some set $F$) has a \emph{$2$-sided layout}  (with respect
to corners $A,B$) if:
\begin{enumerate}
    \item There exists a rectangle $\Theta$ that contains all intersections of curves
    and such that 
	\begin{enumerate}[(i)]
\item the top of $\Theta$ is intersected, from right to left in order, 
    by the curves of the vertices of $P_{AB}$, 
\item the bottom of $\Theta$ is intersected,
	from left to right in order, by the curves of the vertices of $P_{BA}$. 
\end{enumerate}
    \item Any curve $\bb{v}$ of an outer-face vertex $v$ has at most one bend.
    (By 1., this implies that $\bb{A}$ and $\bb{B}$ have no bends.)
\end{enumerate}
\end{definition}

\begin{definition}[$3$-sided layout]
Let $G$ be a $W$-triangulation and $A,B,C$ be three distinct vertices in CCW order on the outer-face of $G$. 
Let $F$ be a set of exactly one outer-face edge incident to $C$.
An \int{F} $B_2$-VPG representation of $G$ has a \emph{$3$-sided layout} 
(with respect to corners $A,B,C$) if:
\begin{enumerate}
    \item There exists a rectangle $\Theta$ containing all intersections of curves 
    so that 
	\begin{enumerate}[(i)]
	\item the top of $\Theta$ is intersected, from right to left in order, 
    	by the curves of the vertices on $P_{AB}$;
	\item the left side of $\Theta$ is intersected, from top to bottom in order, by the
    curves of the vertices on $P_{Bb_{s-1}}$, possibly followed by $\bb{C}$;
\footnote{Recall that $(b_{s-1},C)$ and $(C,c_2)$ are the two incident edges of
$C$ on the outer-face.}
	\item the bottom of $\Theta$ is intersected, from right to left in order, by the curves of vertices 
    on $P_{c_2A}$ in reversed order, possibly followed by $\bb{C}$;%
\addtocounter{footnote}{-1}%
\footnotemark
	\item curve $\bb{b_s} = \bb{C} = \bb{c_1}$ intersects the boundary of $\Theta$ exactly once; it
    is the bottommost curve to intersect the left side of $\Theta$ if the special edge in
   $F$ is $(C,c_2)$, and $\bb{C}$ is
    the leftmost curve to intersect the bottom of $\Theta$ if the special edge in $F$ is $(C,b_{s-1})$.
    	\end{enumerate}
    \item Any curve $\bb{v}$ of an outer-face vertex $v$ has at most one bend.
    (By 1., this implies that $\bb{B}$ has precisely one bend.) 
    \item $\bb{A}$ and $\bb{C}$ have no bends.
\end{enumerate} \end{definition}

We also need the concept of a \emph{reverse $3$-sided layout},
which is similar to the 3-sided layout except that $B$ is straight and $A$ has a bend.  Formally:

\begin{definition}[Reverse $3$-sided layout]
Let $G$ be a connected planar graph and $A,B,C$ be three distinct vertices in CCW order on the outer-face of $G$. 
Let $F$ be a set of exactly one outer-face edge incident to $C$.
An \int{F} $B_2$-VPG representation of $G$ has a \emph{reverse $3$-sided layout} 
(with respect to corners $A,B,C$) if:
\begin{enumerate}
    \item There exists a rectangle $\Theta$ containing all intersections of curves 
    so that 
	\begin{enumerate}[(i)]
	\item the right side of $\Theta$ is intersected, from bottom to top in order, 
    	by the curves of the vertices on $P_{AB}$;
	\item the left side of $\Theta$ is intersected, from top to bottom in order, by the
    curves of the vertices on $P_{Bb_{s-1}}$, possibly followed by $\bb{C}$;
	\item the bottom of $\Theta$ is intersected, from right to left in order, by the curves of vertices 
    on $P_{c_2A}$ in reversed order, possibly followed by $\bb{C}$;%
	\item curve $\bb{b_s} = \bb{C} = \bb{c_1}$ intersects the boundary of $\Theta$ exactly once; it
    is the bottommost curve to intersect the left side of $\Theta$ if the special edge in
   $F$ is $(C,c_2)$, and $\bb{C}$ is
    the leftmost curve to intersect the bottom of $\Theta$ if the special edge in $F$ is $(C,b_{s-1})$.
    	\end{enumerate}
    \item Any curve $\bb{v}$ of an outer-face vertex $v$ has at most one bend.
    (By 1., this implies that $\bb{A}$ has precisely one bend.) 
    \item $\bb{B}$ and $\bb{C}$ have no bends.
\end{enumerate} \end{definition}

\begin{figure}[ht]
\includegraphics[width=\textwidth]{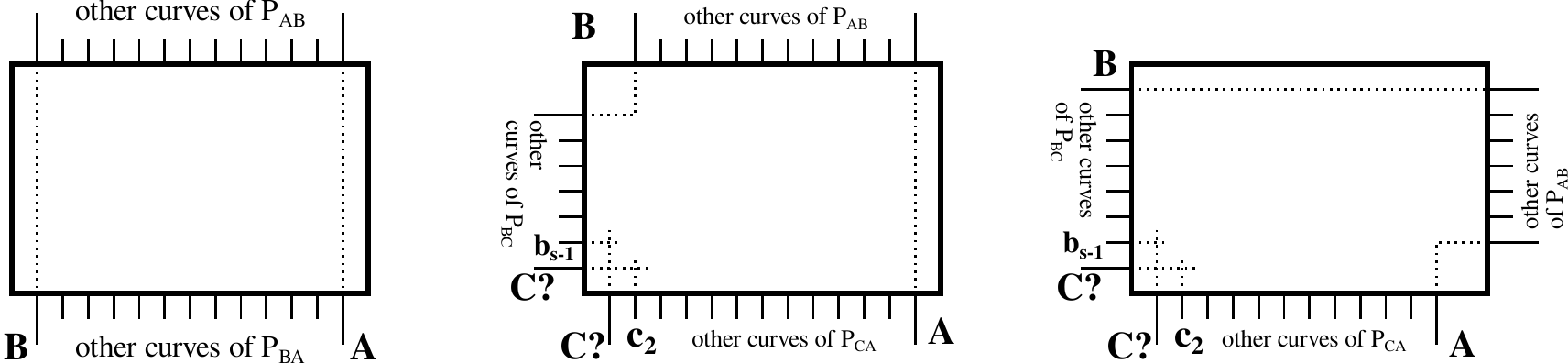}
\caption{Illustration of a 2-sided layout, 3-sided layout, and
reverse 3-sided layout.}
\label{fig:layoutTypes}
\end{figure}


We sometimes refer to the 
rectangle $\Theta$ for these
representations as a~\emph{bounding box}.
Figure~\ref{fig:base-case} (which will serve as base case later) 
shows such layouts for a triangle and varying choices of $F$.

\begin{figure}
	\centering
    \includegraphics[width=\textwidth]{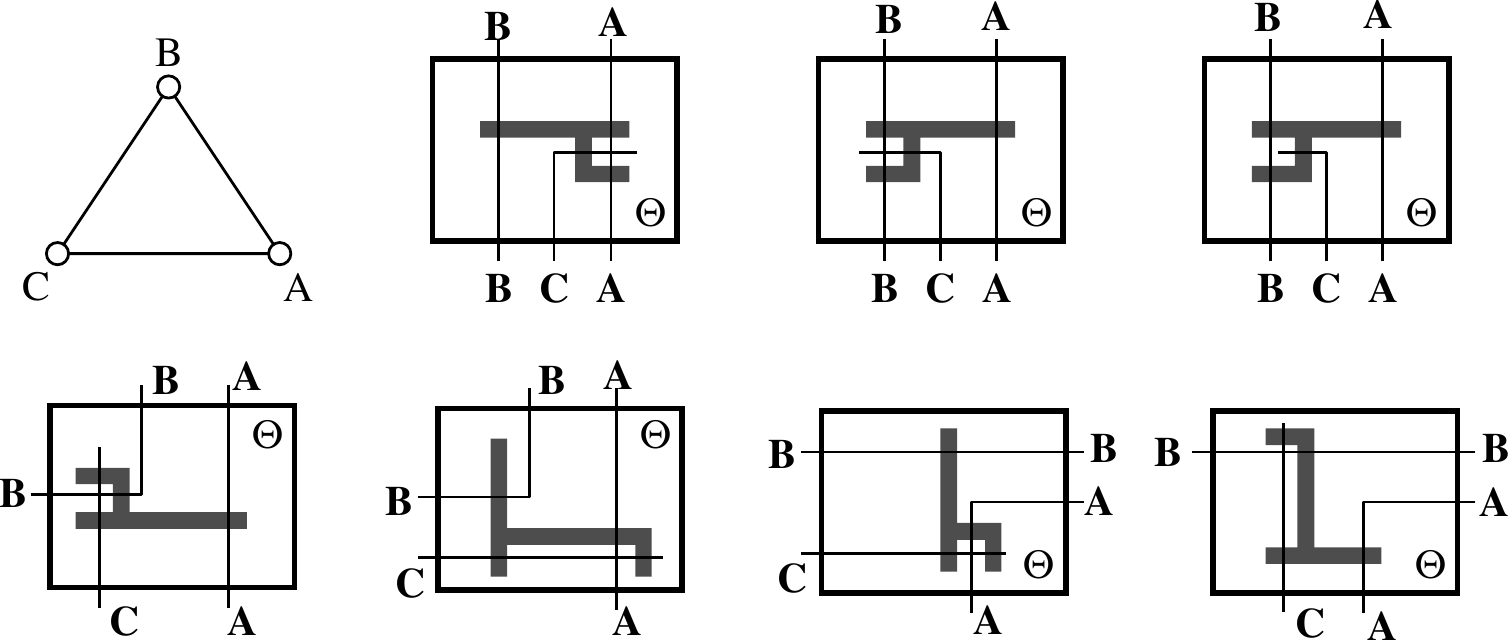}
	\caption{\int{F} representations of a triangle:
(Top) 2-sided representations for $F \in \{{\{(A,C)\}, \{(B,C)\}, \emptyset}\}$. 
(Bottom) 3-sided and reverse 3-sided representations for $F \in \{{\{(A,C)\}, \{(B,C)\}}\}$. 
Private regions are shaded in grey.}
	\label{fig:base-case}
    \label{fig:base-case-3-sided}
\end{figure}

\subsection{Private Regions}

Our proof starts by constructing a representation for 
triangulations without separating triangles.  The construction is then
extended to all triangulations by merging representations of subgraphs
obtained by splitting at separating triangles.  To permit the merge, 
we apply the technique used in~\cite{cit:chalopin-string}
(and also used independently in~\cite{cit:mfcs}):
With every triangular face,
create a region that intersects the curves of vertices of the face in a predefined way 
and does not intersect anything else, specifically not any other such region. Following the notation
of~\cite{cit:mfcs}, we call this a~\emph{private region} (but we use a different shape).

\begin{definition}[Chair-shape]
A \emph{chair-shaped area} is a region bounded by a 10-sided orthogonal polygon
with CW (clockwise) or CCW (counter-clockwise) sequence  of interior angles
$90\degree$, $90\degree$, $270\degree$, $270\degree$, $90\degree$, $90\degree$, $90\degree$, $90\degree$,
$270\degree$, $90\degree$. See also Figure~\ref{fig:private-region}.
\end{definition}

\begin{definition}[Private region]
    Let $G$ be a planar graph with a partial $1$-string $B_2$-VPG representation $R$ and let
    $f$ be a facial triangle in $G$. A \emph{private region} of $f$ 
    is a chair-shaped area $\Phi$ inside $R$ such that:
    \begin{enumerate}
	\item $\Phi$ is intersected by no curves except for the ones representing vertices on $f$.
	\item All the intersections of $R$ are located outside of $\Phi$.
        \item For a suitable labeling of the vertices of $f$ as $\{a,b,c\}$, $\Phi$ is intersected by 
	two segments of $\bb{a}$ and one segment of $\bb{b}$ and $\bb{c}$. The intersections between
	these segments and $\Phi$ occur at the edges of $\Phi$ as depicted in Figure~\ref{fig:private-region}.
    \end{enumerate}
\end{definition}

\begin{figure}
    \centering
    \includegraphics[width=\textwidth]{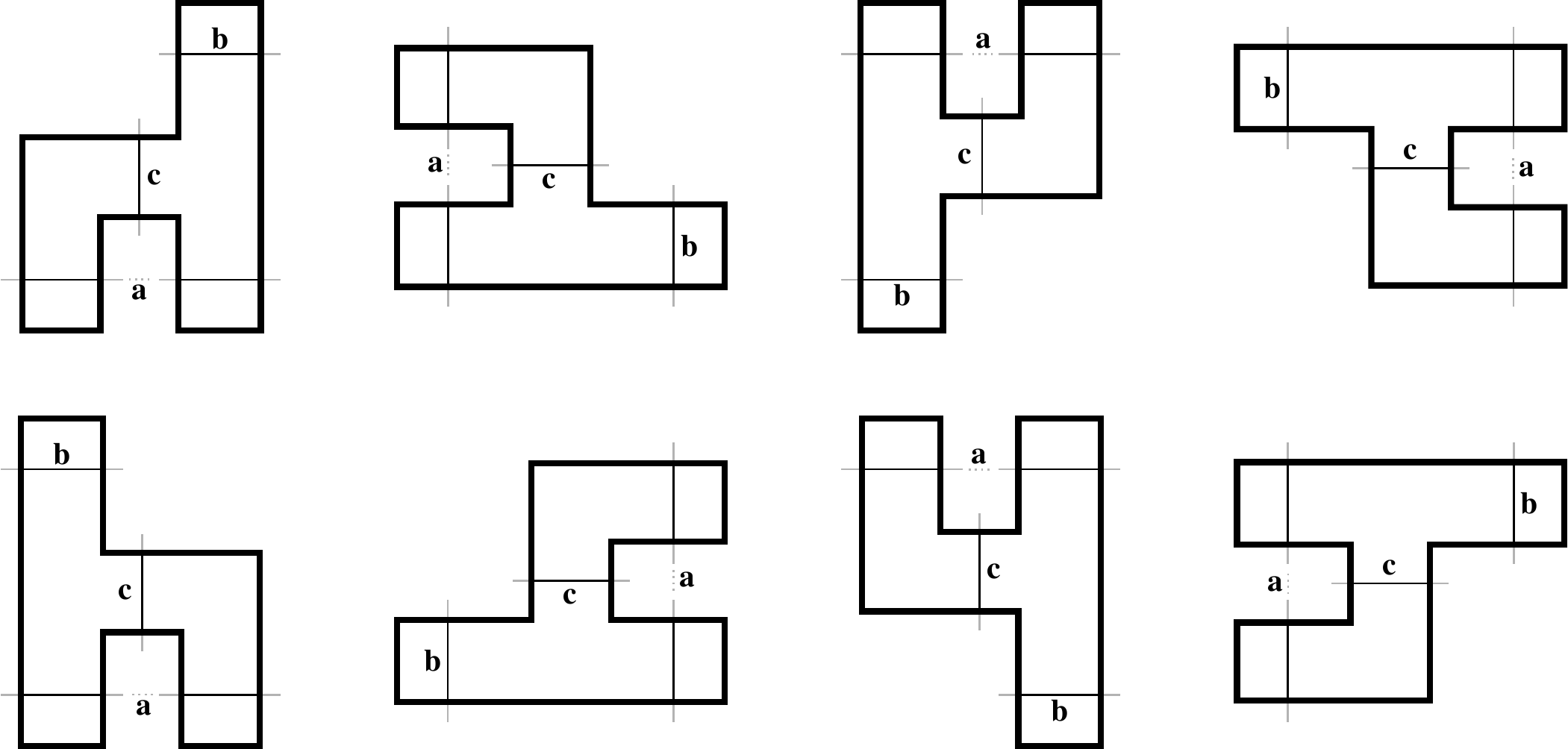}
    \caption{The chair-shaped private region of a triangle $a,b,c$ with possible rotations and flips. Note that labels of $a, b, c$ can be arbitrarily permuted---the curve intersecting the ``base'' of the does not need to be named $\bb{c}$.}
    \label{fig:private-region}
\end{figure}

\subsection{The Tangling Technique}
\label{sec:tangling}

Our constructions will frequently use the following
``tangling technique''. 
Consider a set of $k$ vertical downward rays $\bb{s_1}, \bb{s_2}, \bb{s_3}, \ldots, \bb{s_k}$ placed
beside each other in left to right order. The operation of \emph{bottom-tangling from $\bb{s_1}$ to $\bb{s_k}$ rightwards} 
stands for the following (see also Figure~\ref{fig:tangling}):

\begin{enumerate}
\item For $1 < i \leq k$, stretch $\bb{s_i}$ downwards so that it ends below $\bb{s_{i-1}}$.
\item For $1 \leq i < k$, bend $\bb{s_i}$ rightwards and stretch it so that it crosses $\bb{s_{i+1}}$, but
so that it does not cross $\bb{s_{i+2}}$.
\end{enumerate}

We similarly define right-tangling upwards, top-tangling leftwards and left-tangling downwards as rotation
of bottom-tangling rightwards by 90\degree, 180{\degree} and 270{\degree} CCW. We define
bottom-tangling leftwards as a horizontal flip of bottom-tangling rightwards, and
right-tangling downwards, top-tangling rightwards 
and left-tangling upwards as 90\degree, 180{\degree} and 270{\degree} CCW rotations of 
bottom-tangling leftwards.

\begin{figure}[ht]
    \centering
    \includegraphics[width=.6\textwidth]{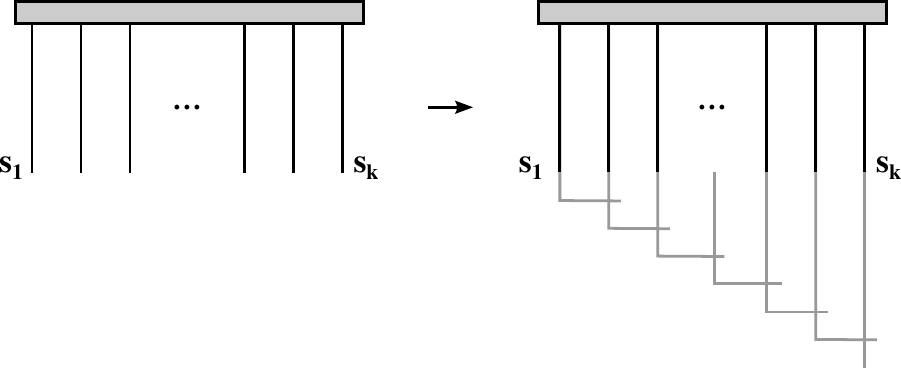}
    \caption{Bottom-tangling from $\bb{s_1}$ to $\bb{s_k}$ rightwards.}
    \label{fig:tangling}
\end{figure}

\section{$2$-Sided Constructions for W-Triangulations}
\label{sec:cz}

We first show the following lemma, which is the key result for
Theorem~\ref{thm:cz}, and will also be used as an ingredient
for the proof of Theorem~\ref{thm:main-claim}.
The reader is also referred to the appendix, where we give an example of a
(3-sided) construction for a graph, which in the recursive cases uses some of the
cases of the proof of Lemma~\ref{lem:2-sided}.

\begin{lemma}
\label{lem:2-sided}

    Let $G$ be a W-triangulation. Let $A, B, C$ be any three corners with respect to which $G$ satisfies the chord condition, and 
    let $F$ be a set of at most one 
	outer-face edge incident to $C$.
    Then $G$ has an \int{F} $1$-string $B_2$-VPG 
    representation with $2$-sided layout with respect to $A,B$.
    Furthermore, this representation has a chair-shaped private
    region for every interior face of~$G$. 
\end{lemma}

We prove Lemma~\ref{lem:2-sided} by induction on the number of vertices. 

First, let us make an observation that will greatly help to reduce
the number of cases of the induction step.
Define $G\rev$ to be the graph obtained from graph $G$ by reversing
the combinatorial embedding, but keeping the same outer-face.  This effectively switches corners
$A$ and $B$, and replaces special edge $(C,c_2)$ by $(C,b_{s-1})$
and vice versa.
If $G$ satisfies the chord condition
with respect to corners $(A,B,C)$, then $G\rev$ satisfies the chord
condition with respect to corners $(B,A,C)$.  (With this new order, the corners
are CCW on the outer-face of $G\rev$, as required.)

Presume we have a 2-sided representation of
$G\rev$.  Then we can obtain a 2-sided representation of $G$ by flipping
the 2-sided one of $G\rev$ horizontally, i.e., along the $y$-axis.    
 Hence for all the following cases, we may (after
possibly applying the above flipping operation) make a restriction on
which edge the special edge is. 

Now we begin the induction. In the base case, $n=3$, so
$G$ is a triangle, and the three corners $A,B,C$ must be the three vertices 
of this triangle.  The desired \int{F} representations 
for all possible choices of $F$
are depicted in Figure~\ref{fig:base-case}.

The induction step for $n \geq 4$ is divided into three cases which we describe
in separate subsections.

\subsection{$C$ has degree $2$} 
\label{case:cz-c-degree-2-2-sided}
Since $G$ is a triangulated disk with $n \geq 4$, 
$(b_{s-1}, c_2)$ is an edge.   Define $G':=G-\{C\}$ and $F':=\{(b_{s-1},c_2)\}$.
We claim that $G'$ satisfies the chord condition for corners $A':=A,
B':=B$ and a suitable choice of $C'\in \{b_{s-1},c_2\}$, and argue this as
follows.  
If $c_2 = A$ or $c_2$ is incident to a chord that ends on $P_{BC}$ other than $(b_{s-1}, c_2)$ (thus $b_{s-1} \neq B$), then 
set $C':=b_{s-1}$. The chord condition holds for $G'$ as $b_{s-1}$ cannot be incident to a chord by planarity and the chord condition for $G$. Otherwise, $c_2$ is not incident to a chord that ends in an interior vertex of $P_{BC}$ other than $b_{s-1}$, so set $C':=c_2$; clearly
the chord condition holds for $G'$.  
Thus in either case, we can apply induction to $G'$.

To create a 2-sided representation of $G$, we use a 2-sided
\int{F'} representation $R'$ of~$G'$ constructed with respect to the aforementioned corners.
We introduce a new vertical curve $\bb{C}$ placed between 
$\bb{b_{s-1}}$ and $\bb{c_2}$ below $R'$. Add a bend at the upper
end of $\bb{C}$ and extend it leftwards or rightwards.   If the special
edge $e$ exists, then extend $\bb{C}$ until it hits the curve of the
other endpoint of $e$; else extend it only far enough to allow for the
creation of the private region. See also Figure~\ref{fig:case1-2sided}.

\begin{figure}[ht]
\includegraphics[width=\textwidth]{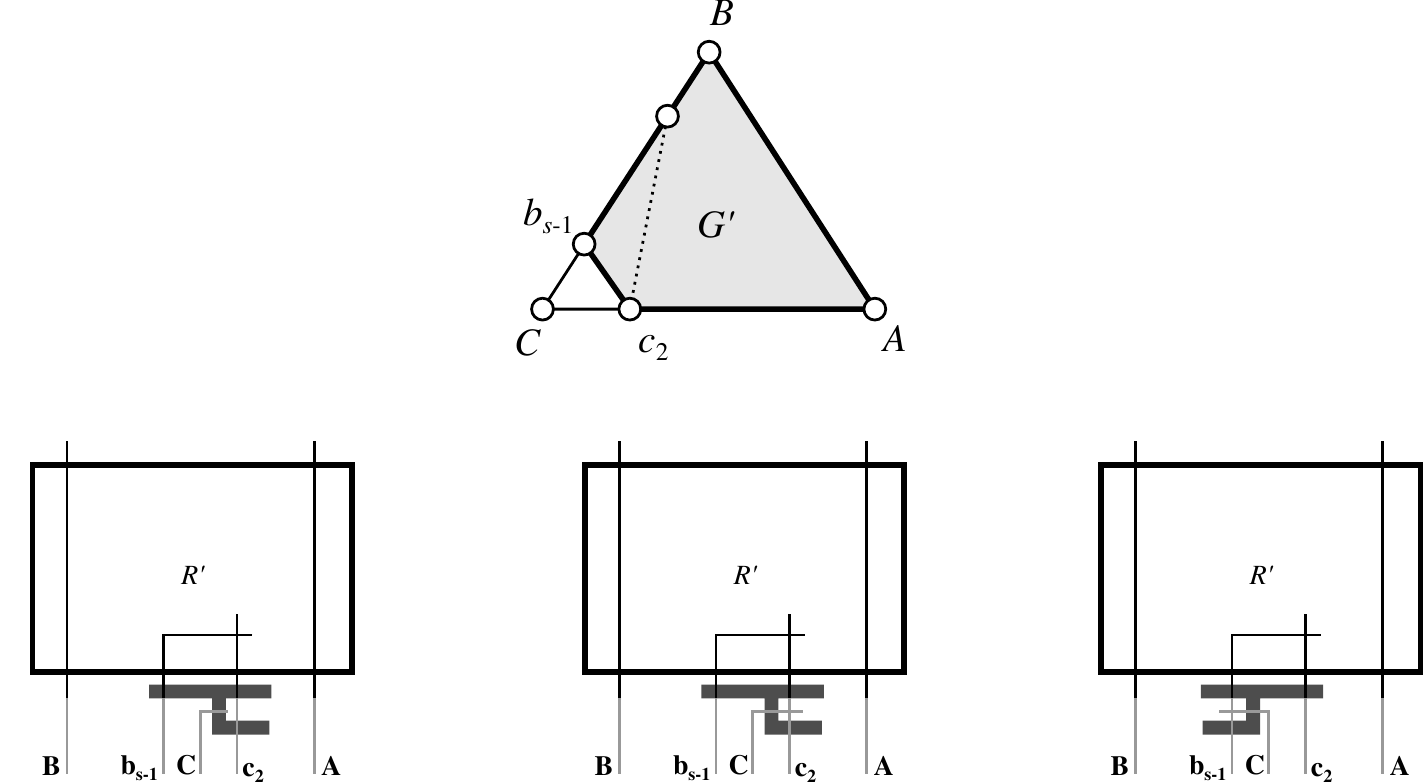}
\caption{Case 1: $2$-sided construction if $C$ has degree 2 and (left) $F = \{\}$, (middle) $F = \{(C,c_2)\}$ and (right) $F = \{(b_{s-1}, C)\}$.} 
\label{fig:case1-2sided}
\end{figure}

\subsection{$G$ has a chord incident to $C$}  
\label{case:cz-special}
We may (after applying the reversal trick)
assume that the special edge, if it exists, is $(C,b_{s-1})$.

By the chord condition, the chord incident to $C$ has 
the form $(C,a_i)$ for some $1 < i < r$.  
The graph $G$ can be split along the chord $(C,a_i)$ into two graphs $G_1$ and
$G_2$.  Both $G_1$ and $G_2$ are bounded by simple cycles, hence they are triangulated
disks.  No edges were added, so neither $G_1$ nor $G_2$
contains a separating triangle.   So both of them are W-triangulations.

We select $(C,A,a_i)$ as corners for $G_1$
and $(a_i,B,C)$ as corners for $G_2$ and can easily verify that
$G_1$ and $G_2$ satisfy the chord condition with respect to those corners: 
\begin{itemize}
\item $G_1$ has no chords on $P_{Aa_i}$ or $P_{CA}$ as they would violate the chord condition in $G$. 
There is no chord on $P_{a_iC}$ as it is a single edge.
\item $G_2$ has no chords on $P_{a_iB}$ or $P_{BC}$ as they would violate the chord condition in $G$. 
There is no chord on $P_{a_iC}$ as it is a single edge.
\end{itemize}

%
%
%
Inductively construct a 2-sided \int{(C,a_i)} representation 
$R_1$ of $G_1$
and a 2-sided  \int{F} representation $R_2$
of $G_2$, both with the aforementioned corners.
Note that $\bb[R_2]{C}$ and $\bb[R_2]{a_i}$ are
on the bottom side
of $R_2$ with $\bb[R_2]{C}$ to the left of~$\bb[R_2]{a_i}$.

Rotate $R_1$ by 180{\degree}, and translate it
so that it is below $R_2$ with $\bb[R_1]{a_i}$ in the same column as
$\bb[R_2]{a_i}$. Stretch $R_1$ and $R_2$ horizontally as needed until $\bb[R_1]{C}$ is 
in the same column as $\bb[R_2]{C}$. Then $\bb[R]{a_i}$ and $\bb[R]{C}$ for $R \in \{R_1,R_2\}$ 
can each be unified without adding bends by adding vertical segments.
The curves of outer-face vertices of $G$ then cross (after suitable lengthening)
the bounding box in the required order. See also Figure~\ref{fig:case-1-3-sided-CZ}.

\begin{figure}
    \centering
    \includegraphics[width=\textwidth]{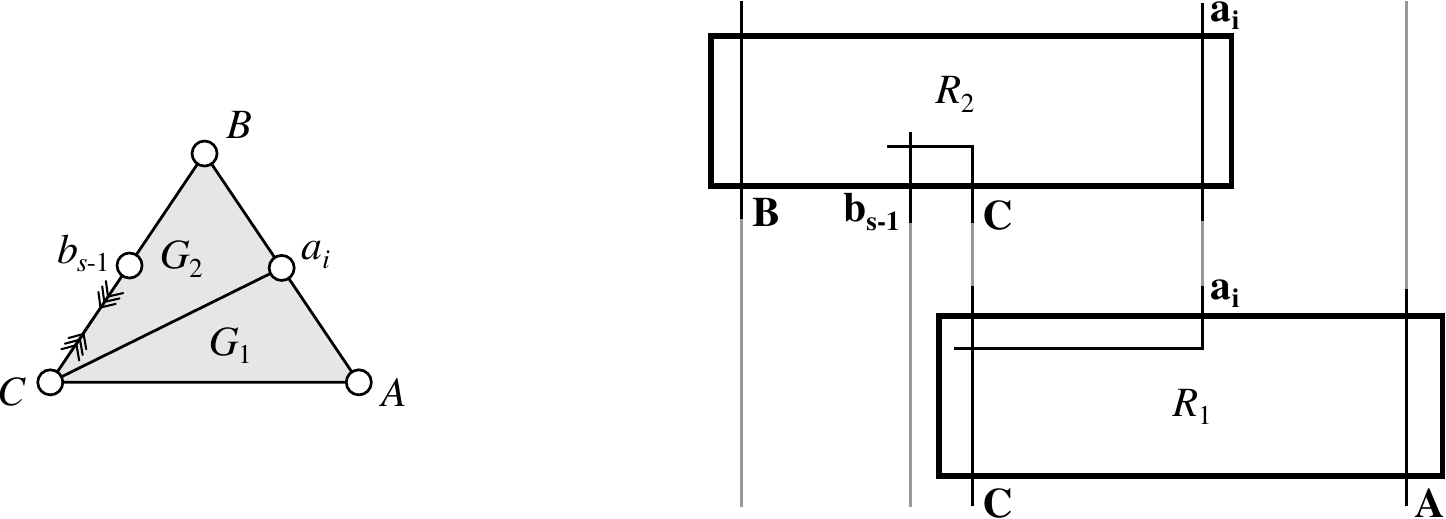}
    \caption{Case~2(a): Constructing an \int{(C,b_{s-1})} 
representation when $C$ is incident to a chord, in 2-sided (middle)
and 3-sided (right) layout. } 
    \label{fig:case-1-3-sided-CZ}
\end{figure}

Every interior face $f$ of $G$ is contained in $G_1$ or $G_2$ and hence has a private region in 
$R_1$ or $R_2$. As our construction does not make
any changes inside the bounding boxes of $R_1$ and $R_2$, the private region of $f$ 
is contained in $R$ as well. 


\subsection{$G$ has no chords incident to $C$ and ${\mathit deg}(C) \geq 3$}
\label{case:cz-chordless}

We may (after applying the reversal trick)
assume that the special edge, if it exists, is $(C,c_2)$.

In this case we split $G$ in a more complicated fashion illustrated
in Figure~\ref{fig:case3-splitting}.
Let $u_1,\dots,u_q$ be the neighbours of vertex $C$ in clockwise order, starting
with $b_{s-1} = u_1$ and ending with $c_{2} = u_q$. We know that $q = \mathit{deg}(C) \geq 3$ and that 
$u_2,\dots,u_{q-1}$ are 
not on the outer-face, since $C$ is not incident to a chord.
Let $u_j$ be a neighbour of $C$ that has at least one neighbour other than $C$ on 
$P_{CA}$, and among all those, choose $j$ to be minimal. Such a
$j$ exists because $G$ is a triangulated disk and 
therefore $u_{q-1}$ is adjacent to both $C$ and $u_q$.   

\begin{figure}[ht]
\centering
\includegraphics[width=.49\textwidth]{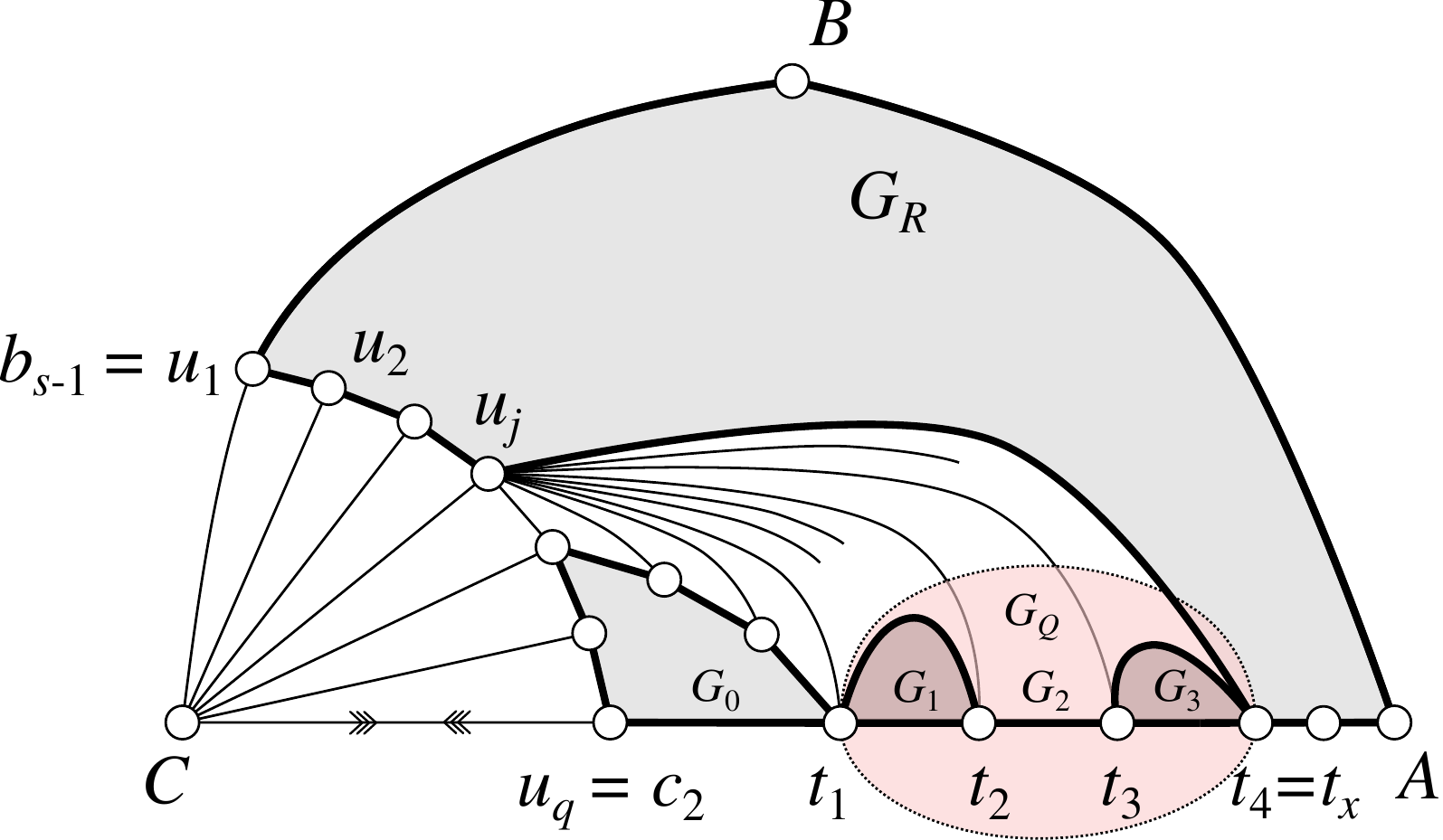}
\includegraphics[width=.49\textwidth]{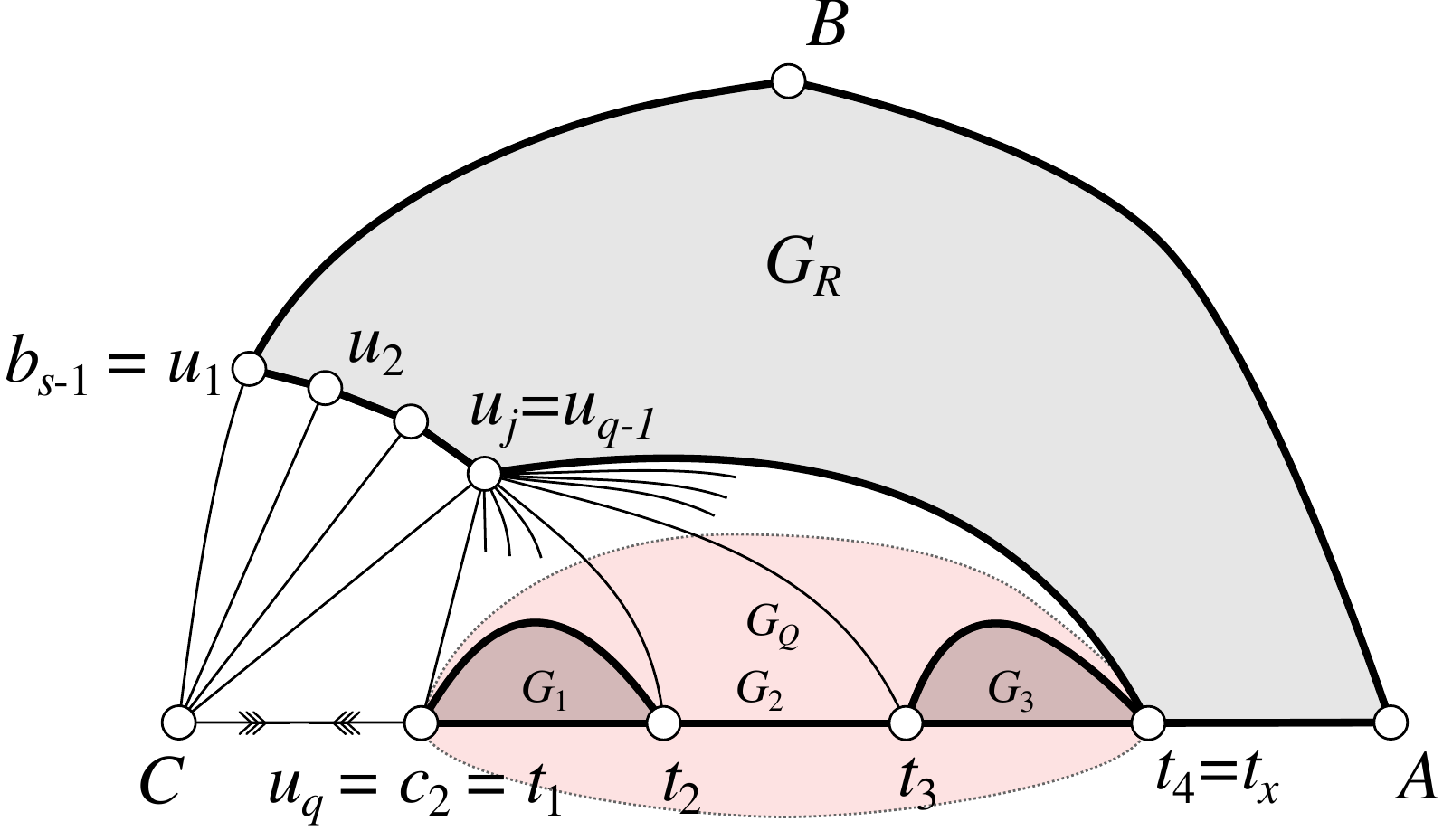}
\caption{Case 3(a): Splitting the graph when $deg(C) \geq 3$, no chord is 
incident to $C$, and $j>1$.
(Left) $j<q-1$; $G_0$ is non-trivial.  (Right) $j=q-1$; $G_0=\{c_2\}$.
} 
\label{fig:case3-splitting}
\end{figure}


We distinguish two sub-cases.

\medskip{\noindent{\bf {Case 3(a): $j\neq 1$. }}}
Denote the neighbours of $u_j$ on $P_{c_2A}$ by $t_1,\dots,t_x$ in 
the order in which they appear on $P_{c_2A}$.
Separate $G$ into subgraphs as follows
(see also Figure~\ref{fig:case3-splitting}):

\begin{itemize}
\item The \emph{right} graph $G_R$ 
is bounded by $(A, \stackrel{P_{AB}}{\ldots}, B, \stackrel{P_{Bu_1}}{\ldots}, u_1,u_2,\dots,u_j, t_x, \stackrel{P_{t_xA}}{\ldots},A)$.
\item Let $G_B$ be the graph bounded by $(u_j, t_1, \stackrel{P_{t_1t_x}}{\ldots}, t_x,u_j)$. 
	We are chiefly interested in its subgraph $G_Q:= G_B-u_j$.
\item Let $G_L$ be the graph bounded by $(C, \stackrel{P_{Ct_1}}{\ldots},t_1,u_j,C)$.
	We are chiefly interested in its subgraph $G_0:=G_L-\{u_j,C\}$.
\end{itemize}

The idea is to obtain representations of these subgraphs and then to combine them suitably.  
We first explain how to obtain the representation $R_R$ used for $G_R$.
Clearly $G_R$ is a W-triangulation, since $u_2,\dots,u_j$ 
are interior vertices of $G$, and hence the outer-face of $G_R$ is a simple cycle.
Set $A_R:=A$ and $B_R:=B$.  If $B \neq u_1$ then set $C_R:=u_1$ and observe that
$G_R$ satisfies the chord
condition with respect to these corners: 
\begin{itemize}
\item $G_R$ does not have any 
chords with both ends on $P_{A_RB_R} = P_{AB}$, $P_{B_Ru_1} \subseteq P_{BC}$, or $P_{t_xA_R} \subseteq P_{CA}$
since $G$ satisfies the chord condition.
\item 
If there were any chords 
between a vertex in $u_1,\dots,u_j$ and a vertex on $P_{C_RA_R}$, then by $C_R=u_1$
the chord would either connect two neighbours of $C$ (hence giving a separating triangle of $G$),
or connect some $u_i$ for $i<j$ to $P_{CA}$ (contradicting the minimality of $j$),
or connect $u_j$ to some other vertex on $P_{t_xA}$ (contradicting that
$t_x$ is the last neighbour of $u_j$ on $P_{CA}$). Hence no such chord can exist either.
\end{itemize}
If $B = u_1$, then set $C_R := u_2$ (which exists by $q\geq 3$)
and similarly verify that it satisfies the chord condition as $P_{B_RC_R}$ is
the edge $(B, u_2)$. Since $C_R\in \{u_1,u_2\}$ in
both cases, we can apply induction on $G_R$ and obtain a $2$-sided
\int{(u_1, u_2)} representation $R_R$ with respect to the aforementioned corners.

Next we obtain a representation for the graph $G_0$, which is bounded by $u_{j+1},\dots,u_q,P_{c_2t_1}$
and the neighbours of $u_j$ in CCW order between $t_1$ and $u_{j+1}$.  We distinguish two cases:

\begin{enumerate}[(1)]
	\item $j=q-1$, and hence $t_1=u_q=c_2$ and $G_0$ consists of only $c_2$.
		In this case, the representation of $R_0$ consists of a single vertical line segment $\bb{c_2}$.
	\item $j<q-1$, so $G_0$ contains at least three vertices $u_{q-1},u_q$ and
		$t_1$.  Then $G_0$ is a W-triangulation since $C$ is not incident
		to a chord and by the choice of $t_1$.  Also, it satisfies the chord condition
		with respect to corners $A_0 := c_2, B_0 := t_1$ and $C_0 := u_{j+1}$ 
		since the three paths on its outer-face are sub-paths of $P_{CA}$ or contained
		in the neighbourhood of $C$ or $u_j$.  In this case,
		construct a 2-sided \int{(u_{j+1},u_{j+2})} representation $R_0$ of $G_0$ with
		respect to these corners inductively.
\end{enumerate}

Finally, we create a representation $R_Q$ of $G_Q$.   If $G_Q$ is a single vertex or
a single edge, then simply use vertical segments for the curves of its vertices (recall that there is no special edge here).  Otherwise, we can show:

\begin{claim}
\label{claim}
$G_Q$ has a $2$-sided \int{\emptyset} $1$-string $B_2$-VPG representation with respect to corners
$t_1$ and $t_x$.
\end{claim}
\begin{proof}
$G_Q$ is not necessarily 2-connected, so we cannot apply induction directly.  
Instead we break it into $x-1$ graphs $G_1,\dots,G_{x-1}$, where for $i=1,\dots,x-1$ graph
$G_i$ is
bounded by $P_{t_it_{i+1}}$ as well as the neighbours of $u_j$ between
$t_i$ and $t_{i+1}$ in CCW order.  Note that $G_i$ is either a single
edge, or it is bounded by a simple cycle since $u_j$ has no neighbours on
$P_{CA}$ between $t_i$ and $t_{i+1}$.  In the latter case, use $B_i:=t_i$, $A_i:=t_{i+1}$,
and $C_i$ an arbitrary third vertex on $P_{t_it_{i+1}} \subseteq P_{CA}$,
which exists since the outer-face of $G_i$ is a simple cycle and $(t_i, t_{i+1}, u_j)$ is not
a separating triangle. 
Observe
that $G_i$ satisfies the chord condition since all paths on the outer-face of $G_i$
are either part of $P_{CA}$ or in the neighbourhood of $u_j$.  Hence by induction
there exists a 2-sided \int{\emptyset} representation $R_i$ of $G_i$ with respect to the corners of $G_i$.  If $G_i$
is a single edge $(t_i,t_{i+1})$, then let $R_i$ consists of two vertical segments
$\bb{t_i}$ and $\bb{t_{i+1}}$.

Since each representation $R_i$ has at its leftmost end a vertical segment $\bb{t_i}$
and at its rightmost end a vertical segment $\bb{t_{i+1}}$, we can combine all
these representations by aligning $\bb[R_i]{t_i}$ and $\bb[R_{i+1}]{t_i}$ horizontally
and filling in the missing segment.  See also Figure~\ref{fig:chain_blocks}.
One easily verifies that the result is a
2-sided \int{\emptyset} representation of $G_Q$.
\end{proof}

\begin{figure}[ht]
    \centering
    \raisebox{-0.5\height}{\includegraphics[width=.35\textwidth]{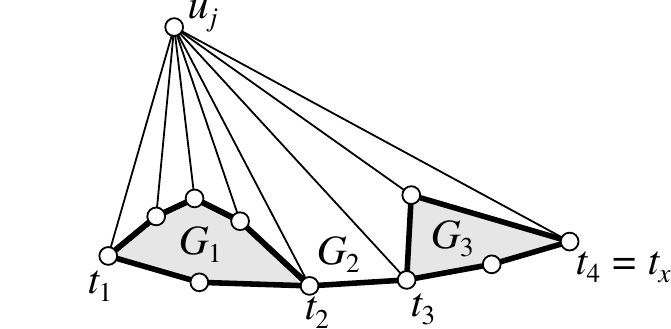}}
\raisebox{-0.5\height}{~~~$\rightarrow$~~~}
    \raisebox{-0.5\height}{\includegraphics[width=.3\textwidth]{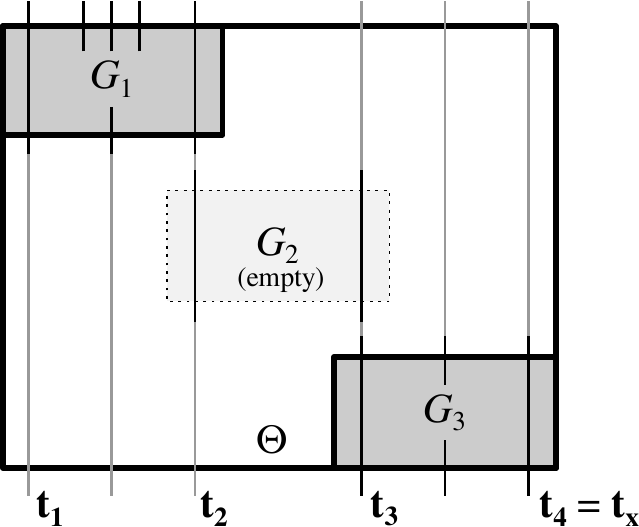}}
    \caption{Left: Graph $G_B$.
        The boundary of $G_Q$ is shown bold. 
        Right: Merging 2-sided \int{\emptyset} representations 
        of $G_i, 1 \leq i \leq 3$, into a 2-sided \int{\emptyset} representation
        of $G_Q$.}
    \label{fig:chain-representation}
    \label{fig:chain_blocks}
\end{figure}

We now explain how to combine these three representations $R_R$, $R_Q$ and $R_0$;
see also Figure~\ref{fig:case-3-CZ}.  
Translate $R_Q$ so that it is below $R_R$ with $\bb[R_R]{t_x}$ and $\bb[R_Q]{t_x}$ in the
same column; then connect these two curves with a vertical segment.
Rotate $R_0$ by 180\degree~and translate it so that it is below $R_R$ and 
to the left and above $R_Q$, and $\bb[R_0]{t_1}$ and $\bb[R_Q]{t_1}$ are in the same column;
then connect these two curves with a vertical segment.
Notice that the vertical segments of $\bb[R_R]{u_2},\dots,\bb[R_R]{u_{j}}$ are at the bottom left of $R_R$.
Horizontally stretch $R_0$ and/or $R_R$ so that $\bb[R_R]{u_2},\dots,\bb[R_R]{u_{j}}$ are to
the left of the vertical segment of $\bb[R_0]{u_{j+1}}$,
but to the right (if $j<q-1$) of the vertical segment of $\bb[R_0]{u_{j+2}}$.
There are such segments by $j > 1$.

Introduce a new horizontal segment $\bb{C}$ and place it so that it
intersects curves $\bb{u_q},\dots,\bb{u_{j+2}},\allowbreak\bb{u_2}, \dots, \bb{u_j},\bb{u_{j+1}}$ (after lengthening them, 
if needed).  Attach a vertical segment to $\bb{C}$. 
If $j<q-1$, then top-tangle $\bb{u_q},\dots,\bb{u_{j+2}}$ rightwards.
(Recall from Section~\ref{sec:tangling} that this creates intersections
among all these curves.)
Bottom-tangle $\bb{u_2},\dots,\bb{u_j}$ rightwards.  
The construction hence creates intersections for all edges in the path $u_1,\dots,u_q$,
except for $(u_{j+2},u_{j+1})$ (which was represented in $R_0$) and
$(u_2,u_1)$ (which was represented in $R_R$).  

Bend and stretch $\bb[R_R]{u_j}$ 
rightwards so that it crosses the curves of all its neighbours in $G_0\cup G_Q$.
Finally, consider the path between the neighbours of $u_j$ CCW from $u_{j+1}$ to $t_x$. 
Top-tangle curves of these vertices rightwards, but omit the intersection
if the edge is on the outer-face (see e.g.~$(t_2,t_3)$ in 
Figure~\ref{fig:case-3-CZ}).

One verifies that the curves intersect the bounding boxes as desired.
The constructed representations contain private regions for all interior faces of $G_R$,
$G_Q$ and $G_0$ by induction. The remaining faces are of the form $(C,u_i,u_{i+1}), 1 \leq i < q$,
and $(u_j, w_k, w_{k+1})$ where $w_{k}$ and $w_{k+1}$ are two consecutive neighbours of $u_j$
on the outer-face of $G_0$ or $G_Q$.
Private regions for those faces are shown in 
Figure~\ref{fig:case-3-CZ}.

\begin{figure}[ht]
\centering
\includegraphics[width=.99\textwidth, trim={0 0 0 13.8cm},clip]{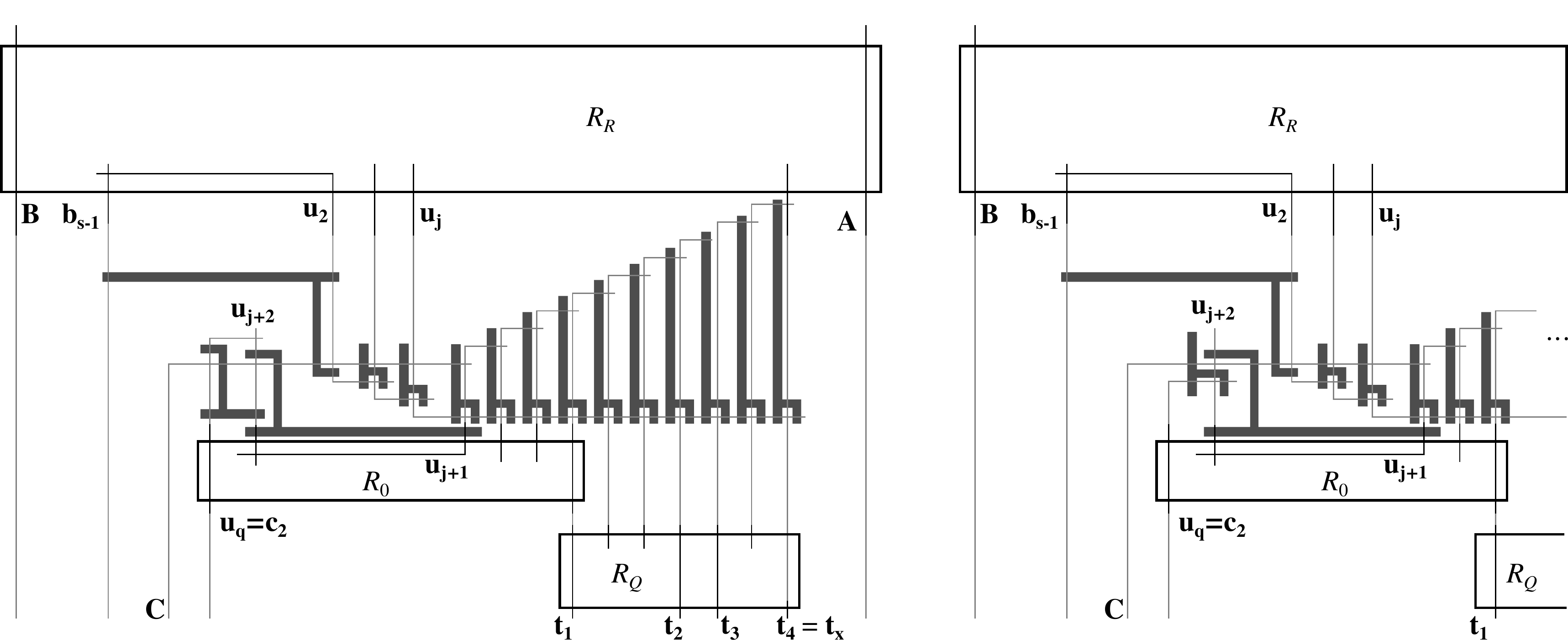}
\caption{Combining subgraphs in Case 3(a).
2-sided 
construction, for $F=\{(C,c_2)\}$ and $F=\emptyset$.
The construction matches the graph depicted in Figure~\ref{fig:case3-splitting} left.
}
\label{fig:case-3-CZ}
\end{figure}

\medskip{\noindent{\bf {Case 3(b): $j=1$, i.e., there exists a chord $(b_{s-1},c_i)$. }}}
In this case we cannot use the above construction directly since we 
need to bend $\bb{u_j}=\bb{u_1}=\bb{b_{s-1}}$ horizontally rightwards to create 
intersections, but then it no longer extends vertically downwards as required
for $\bb{b_{s-1}}$.   Instead we use a different construction.

Edge $(b_{s-1},c_i)$ is a chord from $P_{BC}$ to $P_{CA}$.  Let $(b_{k},c_{\ell})$
be a chord from $P_{BC}$ to $P_{CA}$ that maximizes $k-\ell$, i.e., is
furthest from $C$ (our construction in this case
actually works for any chord from $P_{BC}$ to $P_{CA}$---it is not necessary that $k = s-1$). Note that possibly $\ell=t$ (i.e., the chord is incident
to $A$) or $k=1$ (i.e., the chord is incident to $B$), but not both by the
chord condition.   We assume here that $\ell<t$, the other case is symmetric.

\begin{figure}[ht]
\begin{center}
\includegraphics[width=.25\textwidth]{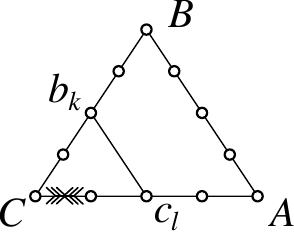}\hspace{2em}
\includegraphics[width=.25\textwidth]{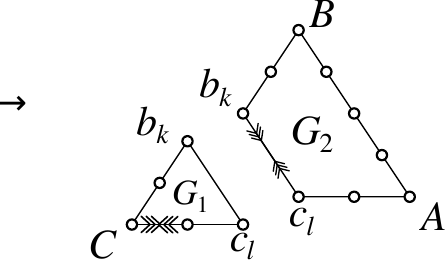}\hspace{3em}
\includegraphics[width=.25\textwidth]{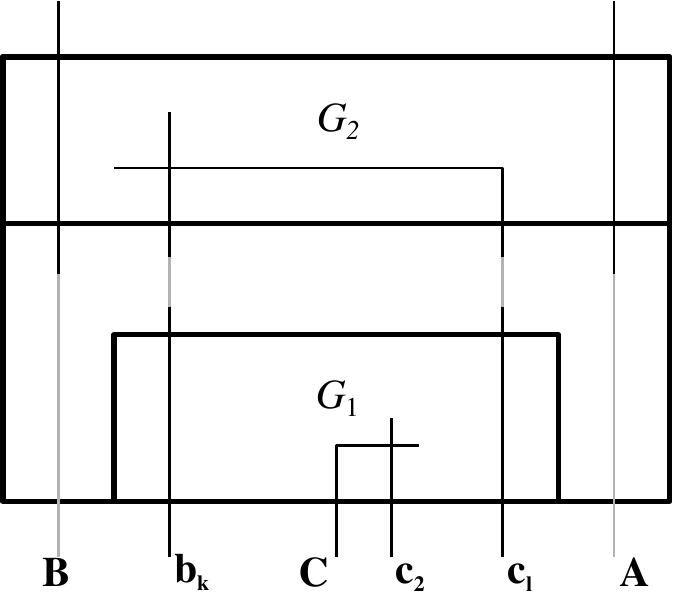}
\end{center}
\caption{Case~3(b): Construction of a 2-sided \int{(C,c_2)} representation of $G$
    with a chord $(b_k,c_\ell)$.}
\label{fig:cz-simpler}
\end{figure}


In order to construct a 2-sided \int{F} representation of $G$, split the graph 
along $(b_{k},c_{\ell})$ into two W-triangulations $G_1$ (which includes $C$
and the special edge, if any) and
$G_2$ (which includes $A$).  Set
 $(A,B,c_{\ell})$ as corners for $G_1$ (these are three distinct vertices
by $c_\ell\neq A$) and set $(c_{\ell},b_{k},C)$ as corners for $G_2$
and verify the chord condition:

\begin{itemize}
    \item $G_1$ has no chords on either $P_{Cc_{\ell}} \subseteq P_{CA}$ or $P_{b_{k}C} \subseteq P_{BC}$
        as they would contradict the chord condition in $G$.
        The third side is a single edge $(b_{k},c_{\ell})$ and so it does not have any chords either. 
    \item $G_2$ has no chords on either $P_{c_{\ell}A} \subseteq P_{CA}$ or $P_{AB}$ as they
    would violate the chord condition in $G$. It does not have any chords on the path $P_{Bc_{\ell}}$ due to 
    the selection of the chord $(b_{k},c_{\ell})$ and by the chord condition in $G$ and by the chord condition in $G$. 
\end{itemize}

Thus, by induction, $G_1$ has a 2-sided \int{F} representation $R_1$
and $G_2$ has a 2-sided \int{(b_{k},c_{\ell})} representation $R_2$ with respect to the aforementioned corners. 
Translate and horizontally stretch $R_1$ and/or $R_2$ so that
$\bb{b_k^{R_1}}$ and $\bb{c_{\ell}^{R_1}}$ are aligned with
$\bb{b_k^{R_2}}$ and $\bb{c_{\ell}^{R_2}}$, respectively, and connect
each pair of curves with a vertical segment.
Since $\bb{b_{k}^{R_1}}$ and $\bb{c_{\ell}^{R_1}}$ have no bends, this 
does not increase the number of bends on any curve and produces 
a 2-sided \int{F} representation of $G$. 
All the faces in $G$ have a private region inside one of the representations of $G_1$ or $G_2$.

\bigskip
This ends the description of the construction in all cases, and
hence proves Lemma~\ref{lem:2-sided}. We now show how Lemma~\ref{lem:2-sided} implies Theorem~\ref{thm:cz}:

\begin{proof}[Proof of Theorem~\ref{thm:cz}]
Let $G$ be a 4-connected planar graph.  Assume first that $G$ is triangulated,
which means that it is a $W$-triangulation.  Let $(A,B,C)$ be the outer-face
vertices and start with an \int{(B,C)}-representation of $G$ (with respect to corners $(A,B,C)$) that exists by
Lemma~\ref{lem:2-sided}.
The intersections of the other two outer-face edges $(A,C)$ and $(A,B)$ 
can be created by tangling $\bb{B},\bb{A}$ and
$\bb{C},\bb{A}$ suitably (see Figure~\ref{fig:completion-CZ}).

Theorem~\ref{thm:cz} also stipulates that every curve used in a representation has at most one vertical segment. This is true for all curves added during the construction. Furthermore, we join two copies of a curve only by aligning
and connecting their vertical ends, so all curves have at most one
vertical segment. 

This proves Theorem~\ref{thm:cz} for 4-connected triangulations.  To handle
an arbitrary 4-connected planar graph, \emph{stellate} the graph, i.e.,
insert into each non-triangular face $f$
a new vertex $v$ and connect it to all vertices on $f$.
By 4-connectivity this creates no
separating triangle and the graph is triangulated afterwards.  Finding a
representation of the resulting graph and deleting the curves of all added
vertices yields the result.
\end{proof}

\begin{figure}[ht]
\centering
\includegraphics[width=.30\textwidth]{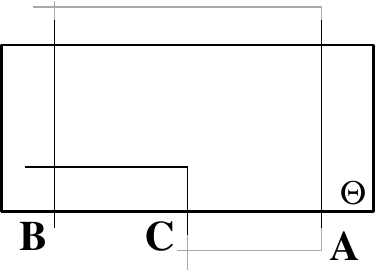}
\caption{Completing a 2-sided \int{(B,C)} representation 
by adding intersections for $(A,B)$ and $(A,C)$.}
\label{fig:completion-CZ}
\end{figure}

\section{3-Sided Constructions for W-Triangulations}
\label{sec:Wtriangulations}

Our key tool for proving Theorem~\ref{thm:main-claim} is the following lemma:

\begin{lemma}
\label{lem:3-sided}
    Let $G$ be a W-triangulation and let $A,B,C$ be any three corners with respect to which $G$ satisfies the chord condition. For any $e \in \{(C,b_{s-1}), (C, c_2)\}$,
    $G$ has an \int{e} $1$-string $B_2$-VPG representation with $3$-sided layout 
    and an \int{e} $1$-string $B_2$-VPG representation with reverse $3$-sided layout.
    Both representations have a chair-shaped private region for every 
	interior face. 
\end{lemma}

The proof of Lemma~\ref{lem:3-sided} will use induction on the number of vertices. To combine the representations
of subgraphs, we sometimes need them to have a 2-sided layout, and hence we frequently use Lemma~\ref{lem:2-sided} proved in Section~\ref{sec:cz}. Also, notice that for Lemma~\ref{lem:3-sided} the special edge \emph{must}
exist (this is needed in Case 1 to find private regions), while for
Lemma~\ref{lem:2-sided}, $F$ is allowed to be empty.

We again reduce the number of cases 
in the proof of Lemma~\ref{lem:3-sided} 
by using the reversal trick.  Define
$G\rev$ as in Section~\ref{sec:cz}.
Presume we have a 3-sided/reverse 3-sided representation of
$G\rev$.    
We can obtain a 3-sided/reverse 3-sided representation of $G$ by flipping the 
reverse 3-sided/3-sided representation of $G\rev$ diagonally (i.e., along the line
defined by $(x=y)$).  Again, this effectively switches corners
$A$ and $B$ (corner $C$ remains the same), and replaces special edge $(C,c_2)$ by $(C,b_{s-1})$
and vice versa.
If $G$ satisfies the chord condition
with respect to corners $(A,B,C)$, then $G\rev$ satisfies the chord condition with respect to corners $(B,A,C)$.
Hence for all the following cases, we may again (after
possibly applying the above flipping operation) make a restriction on
which edge the special edge is. Alternatively, we only need to give the 
construction
for the 3-sided, but not for the reverse 3-sided layout. 

So let $G$ and a special edge $e$ be given, and set $F=\{e\}$.
In the base case, $n=3$, so
$G$ is a triangle, and the three corners $A,B,C$ must be the three vertices 
of this triangle.  The desired \int{F} representations 
for all possible choices of $F$
are depicted in Figure~\ref{fig:base-case}.
The induction step for $n \geq 4$ uses the same case distinctions
as the proof of Lemma~\ref{lem:2-sided}; we describe
these cases in separate subsections.


\subsection{$C$ has degree $2$} 
\label{case:c-degree-2-2-sided}
Since $G$ is a triangulated disk with $n \geq 4$, 
$(b_{s-1}, c_2)$ is an edge.   Define $G'$ as in Section~\ref{case:cz-c-degree-2-2-sided} to be $G-\{C\}$ and recall that $G'$ satisfies the chord condition for corners $A':=A,
B':=B$ and a suitable choice of $C'\in \{b_{s-1},c_2\}$
Thus, we can apply induction to $G'$.


To create a 3-sided representation of $G$, we use a 3-sided
\int{F'} representation $R'$ of $G'$, where $F'=\{(b_{s-1},c_2)\}$.
Note that regardless of which vertex
is $C'$, we have $\bb{b_{s-1}}$ as bottommost curve on the left and
$\bb{c_2}$ as leftmost curve on the bottom.
Introduce a new horizontal segment representing ${C}$ which intersects $\bb{c_2}$ if $F = \{(C,c_2)\}$,
or a vertical segment which intersects $\bb{b_{s-1}}$ if $F = \{(C,b_{s-1})\}$.

After suitable lengthening, the curves intersect the 
bounding box in the required order.
One can find the chair-shaped
private region for the only new face $\{C,c_2,b_{s-1}\}$
as shown in 
Figure~\ref{fig:case2-simpler}.
Observe that no bends were added to the curves of $R'$ and that $C$
has the required number of bends.

\begin{figure}[ht]
\centering
\includegraphics[width=\textwidth, trim={0 4cm 0 0}, clip]{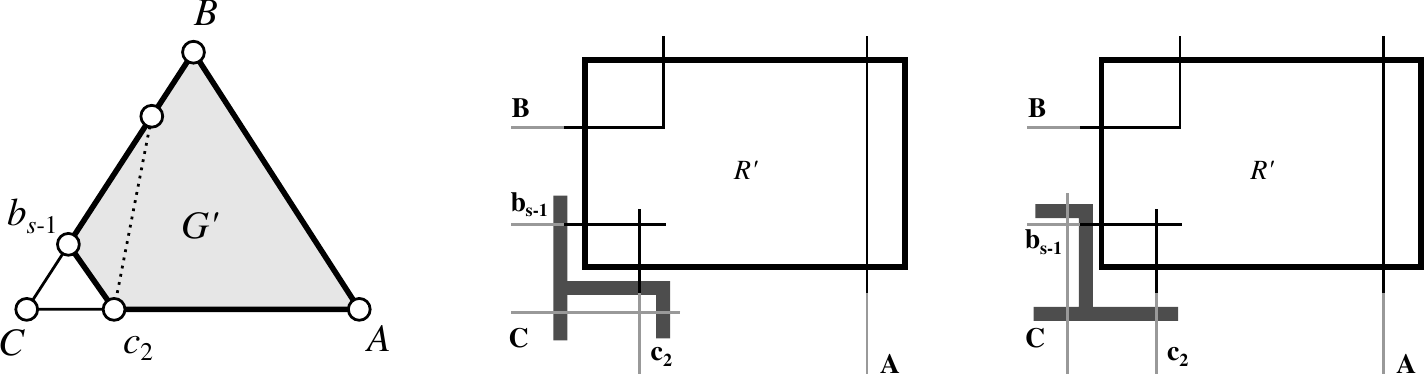}
\caption{Case 1: 3-sided representation if $C$ has degree $2$.
}  
\label{fig:case2-simpler}
\end{figure}

Since we have given the constructions for both possible special edges,
we can obtain the reverse 3-sided representation by diagonally
flipping a 3-sided representation of $G\rev$.



\subsection{$G$ has a chord incident to $C$}  
\label{case:special}

Let $(C,a_i)$ be a chord that minimizes $i$ (i.e., is closest to $A$).
Define W-triangulations $G_1$ and $G_2$ with corners $(C,A,a_i)$ for $G_1$ and $(a_i,B,C)$ for $G_2$ as in Section~\ref{case:cz-special}, and recall that they satisfy the chord condition.
So, we can apply induction to both $G_1$ and $G_2$, obtain representations
$R_1$ and $R_2$ (with respect to the aforementioned corners) for them, and combine them suitably.  We will do so
for both possible choices of special edge, and hence need not give the
constructions for reverse 3-sided layout due to the reversal trick.

{\medskip\noindent{\bf Case 2(a): $F=\{(C,b_{s-1})\}$. }}
Using Lemma~\ref{lem:2-sided}, construct a 2-sided \int{(C,a_i)} representation $R_1$ of $G_1$ with respect to the aforementioned corners of $G_1$.
Inductively, construct a 3-sided \int{F} representation $R_2$
of $G_2$ with respect to the corners of $G_2$. Note that $\bb[R_2]{C}$ and $\bb[R_2]{a_i}$ are
on the bottom side of $R_2$ with $\bb[R_2]{C}$ to the left of~$\bb[R_2]{a_i}$.

First, rotate $R_1$ by 180{\degree}. We can now merge
$R_1$ and $R_2$ as described in Section~\ref{case:cz-c-degree-2-2-sided}
since all relevant curves end vertically in $R_1$ and $R_2$. 
The curves of outer-face vertices of $G$ then cross (after suitable lengthening) the bounding box in the required order. See also Figure~\ref{fig:case-1-3-sided}.

\begin{figure}[ht]
    \centering
    \includegraphics[width=\textwidth]{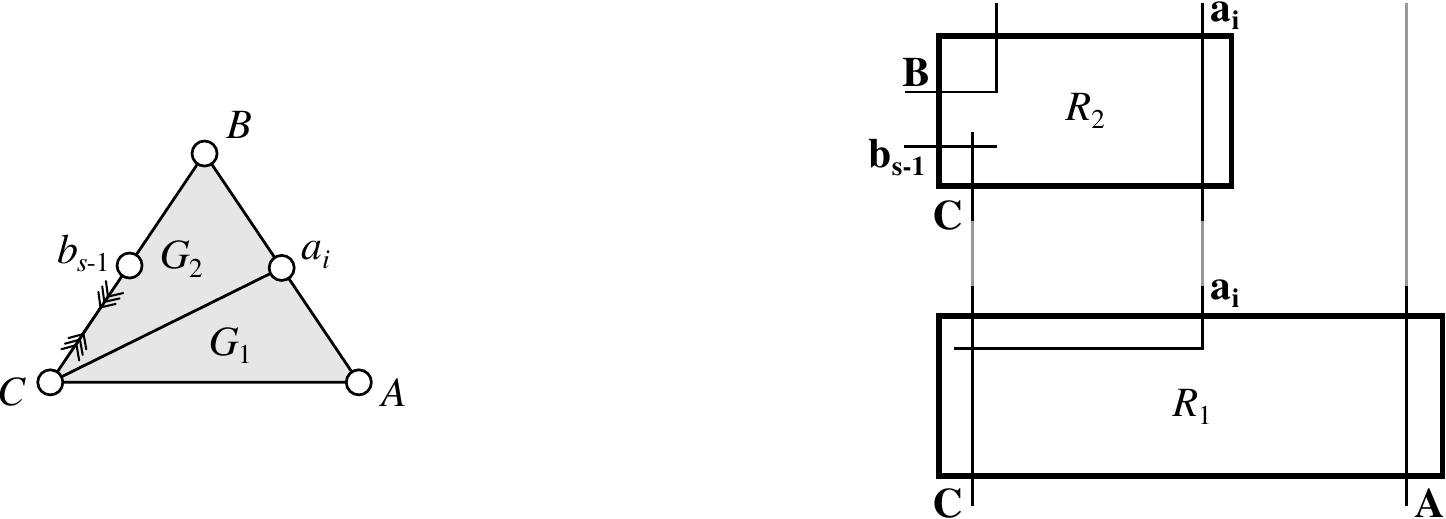}
    \caption{Case~2(a): Constructing a 3-sided \int{(C,b_{s-1})} 
representation when $C$ is incident to a chord.}
    \label{fig:case-1-3-sided}
\end{figure}


\medskip{\noindent{\bf {Case 2(b): $F = \{(C,c_2)\}$. }}}
For the 3-sided construction, it does not seem possible to merge suitable
representations of $G_1$ and $G_2$ directly, since the geometric restrictions
imposed onto curves $\bb{A}, \bb{B}, \bb{C}, \bb{c_2}$ and $\bb{a_i}$ by the
3-sided layout cannot be satisfied using 3-sided and 2-sided representations
of $G_1$ and $G_2$.  We hence use an entirely different approach that splits
the graph further; it resembles
Case~1 in \cite[Proof of Lemma~2]{cit:ham-cycle}. 
Let $G_Q=G_1-C$, and observe that it is bounded by $P_{c_2A}$, $P_{A,a_i}$,
and the path formed by the neighbours $c_2 = u_1, u_2, \ldots, u_q = a_i$ of $C$ in $G_1$ in CCW order. 
We must have $q\geq 2$, but possibly $G_1$ is a triangle $\{C,A,a_i\}$ and $G_Q$ then
degenerates into an edge.  If $G_Q$ contains at least three vertices, then
$u_2,\dots,u_{q-1}$ are interior since chord $(C,a_i)$ was chosen closest to $A$,
and so $G_Q$ is a W-triangulation.

We divide the proof into two subcases,
depending on whether $A\neq c_2$ or $A=c_2$. See also Figures~\ref{fig:3sided-case1ba}
and~\ref{fig:3sided-case1bb}.

\medskip{\noindent{\bf {Case 2(b)1: $A \neq c_2$. }}}
Select the corners of $G_Q$ as $(A_Q := c_2, B_Q := A, C_Q := a_i = u_q)$, and observe
that it satisfies the chord condition since the three corners are distinct and
the three outer-face paths are
sub-paths of $P_{CA}$ and $P_{AB}$ or in the neighbourhood of $C$, respectively.
Apply Lemma~\ref{lem:2-sided} to construct a 2-sided \int{(u_q,u_{q-1})} representation
$R_Q$ of $G_Q$ with respect to the corners of $G_Q$.
Inductively, construct a 3-sided \int{(C,a_i)} representation $R_2$
of $G_2$ with respect to the corners of $G_2$. 

To combine $R_Q$ with $R_2$,
rotate $R_Q$ by 180{\degree}. 
Appropriately stretch 
$R_Q$ and translate it so that it is below $R_2$ with $\bb[R_Q]{a_i}$ and $\bb[R_2]{a_i}$ in the same
column, and so that the vertical segment of each of the curves $\bb{u_{q-1}}, \ldots, \bb{u_{1}} = \bb{c_2}$ is to the left of the
bounding box of $R_2$. Then $\bb[R_Q]{a_i}$ and $\bb[R_2]{a_i}$
can be unified without adding bends by adding a vertical segment.
Curves $\bb{u_{q-1}}, \ldots, \bb{u_{1}} = \bb{c_2}$ in the rotated $R_Q$ can be
appropriately stretched upwards, intersected by $\bb[R_2]{C}$ after stretching it leftwards, 
and then top-tangled leftwards.  
All the curves of outer-face vertices of $G$ then cross 
(after suitable lengthening) a bounding box in the required order. 

All faces in $G$ that are not interior to $G_Q$ or $G_2$ are bounded by $(C, u_k, u_{k+1})$, $1 \leq k < q$.
The chair-shaped private regions for such faces can be found as shown in Figure~\ref{fig:3sided-case1ba}.

\medskip{\noindent{\bf {Case 2(b)2: $A = c_2$. }}}
In this case the previous construction cannot be applied since the
corners for $G_Q$ would not be distinct.  We give an entirely 
different construction.

If $G_Q$ has at least 3 vertices, then $q\geq 3$ since
otherwise by $A=c_2=u_1$ edge $(A,u_q)$ would be a chord on $P_{AB}$.  Choose as corners
for $G_Q$ the vertices $A_Q:=A, B_Q:=a_i=u_q$ and $C_Q:=u_{q-1}$ and
observe that the chord condition holds since all three paths on the
outer-face belong to $P_{AB}$ or are in the neighbourhood of $C$.
By Lemma~\ref{lem:2-sided}, $G_Q$ has a 2-sided \int{(u_q,u_{q-1})} representation
$R_Q$ with the respective corners and private 
region for every interior face of $G_Q$.
If $G_Q$ has at most 2 vertices, then $G_Q$ consists of edge $(A,a_2)$ only,
and we use as representation $R_2$ two parallel vertical segments
$\bb{a_2}$ and $\bb{A}$.

We combine $R_Q$ with a representation $R_1$ of $G_1$ that is {\em different}
from the one used in the previous cases; in particular we rotate corners.
Construct a reverse 3-sided
layout $R_2$ of $G_2$ with respect to corners $C_2:=a_i$, $A_2:=B$ and $B_2:=C$.   Rotate $R_2$ by 180\degree, and translate it
so that it is situated below $R_Q$ with $\bb[R_Q]{a_i}$ and $\bb[R_2]{a_i}$ 
in the same column. Then, extend $\bb[R_2]{C}$ until it crosses
$\bb[R_Q]{u_{q-1}},\dots, \bb[R_Q]{u_{1}}$ (after suitable lengthening),
and then bottom-tangle 
$\bb[R_Q]{u_{q-1}},\dots, \bb[R_Q]{u_{1}}$ rightwards.  This creates
intersections for all edges in path $u_q,u_{q-1},\dots,u_1$, except
for $(u_q,u_{q-1})$, which is either on the outer-face (if $q=2$) or had an intersection in $R_Q$.
One easily verifies that the result is a 3-sided layout, and private
regions can be found for the new interior faces as shown in Figure~\ref{fig:3sided-case1bb}.

\begin{figure}[ht]
\centering
\includegraphics[width=.8\textwidth]{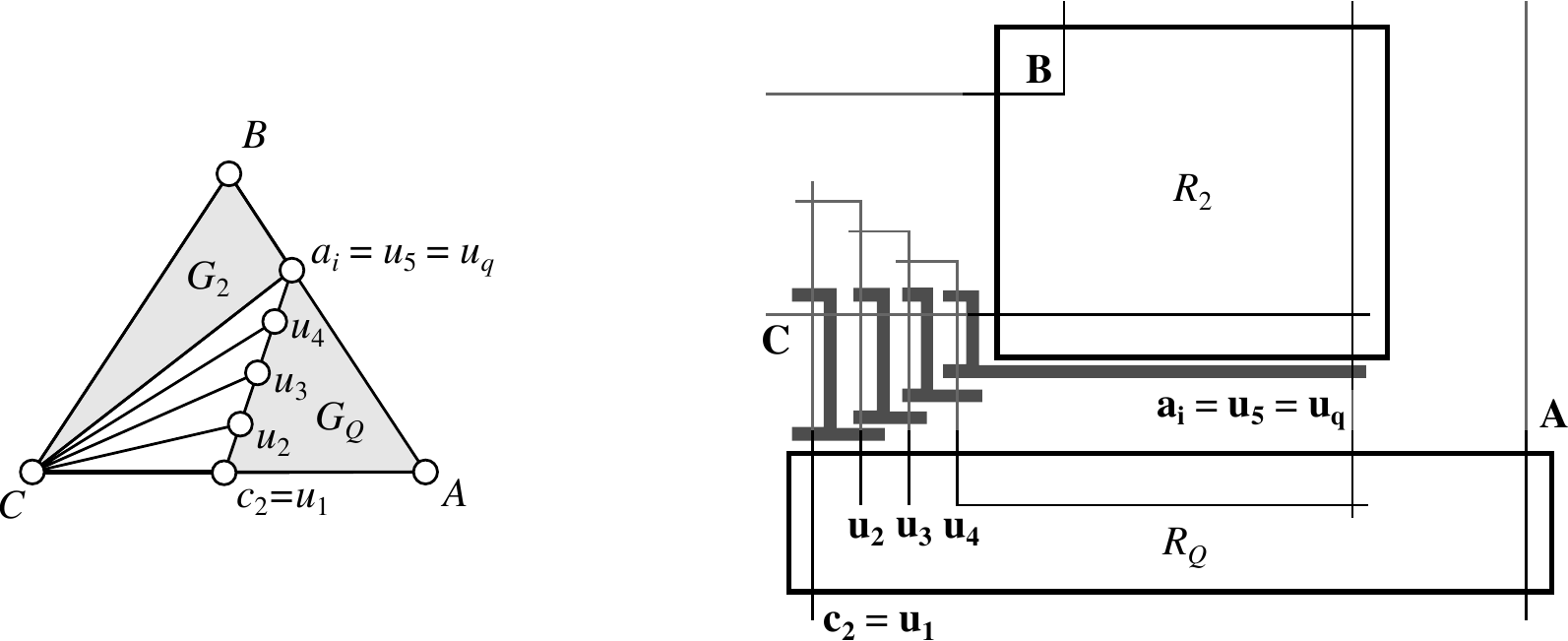}
\caption{Case 2(b)1: $C$ is incident to a~chord, $F = \{(C,c_2)\}$, and $c_2 \neq A$.}
\label{fig:3sided-case1ba}
\end{figure}


\begin{figure}[ht]
\centering
\includegraphics[width=.9\textwidth]{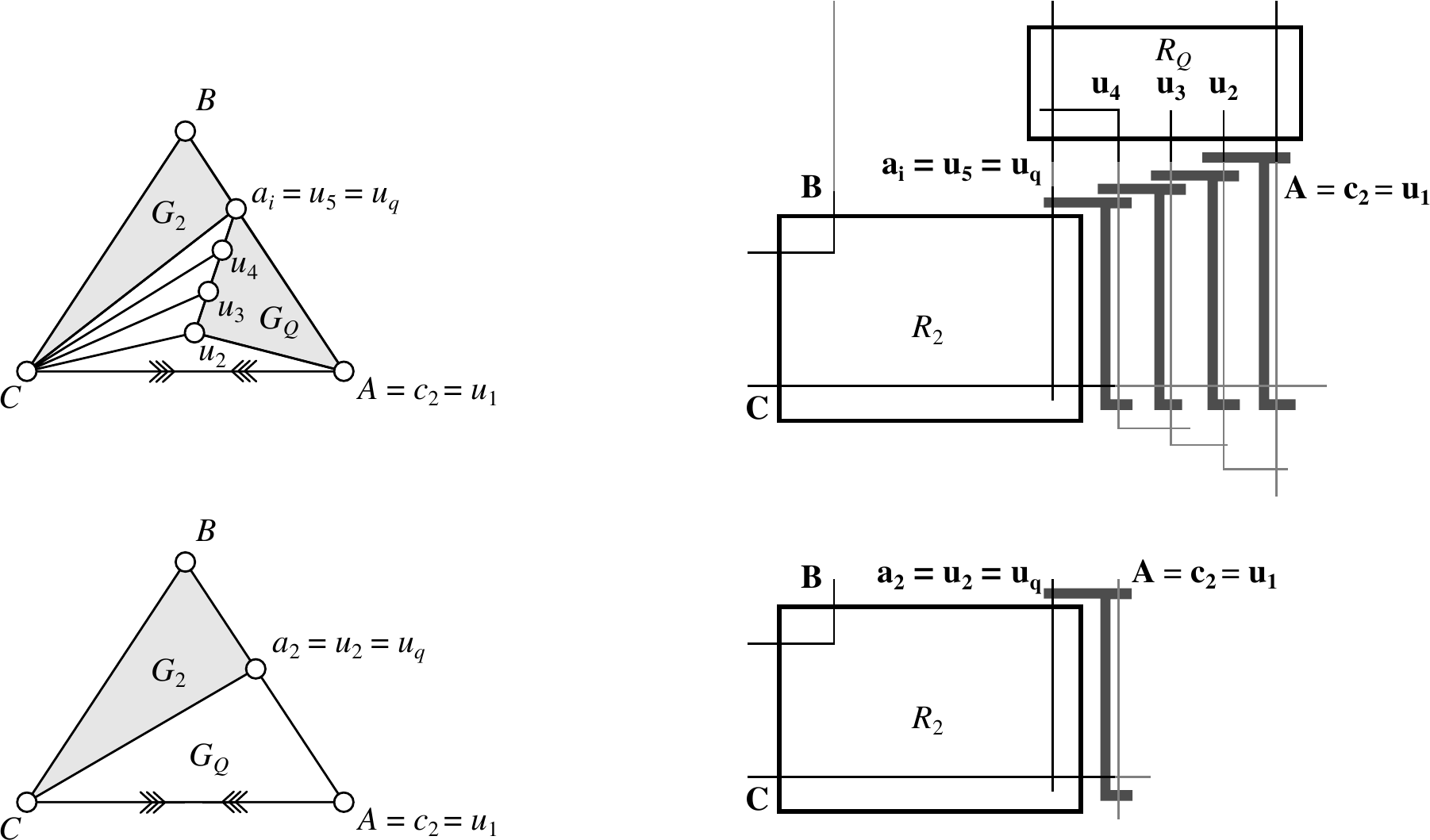}
\caption{Case 2(b)2: Construction when $C$ is incident to a chord, $c_2 = A$, $F = \{(C,c_2)\}$ and $(A,a_i,C)$ is not a face (top),
or when $(A,a_i,C)$ is a face (bottom).  }
\label{fig:3sided-case1bb}
\end{figure}


\subsection{$G$ has no chords incident to $C$ and ${\mathit deg}(C) \geq 3$}
\label{case:chordless}

We will give explicit constructions for 3-sided
and reverse 3-sided layout, and may hence (after applying the reversal trick)
assume that the special edge is $(C,c_2)$.

As in Section~\ref{case:cz-chordless}, let
$u_1,\dots,u_q$  be the neighbours of $C$ and let $j$ be minimal such
that $u_j$ has another neighbour on $P_{AC}$.
We again distinguish two sub-cases.

\medskip{\noindent{\bf {Case 3(a): $j\neq 1$. }}}
As in Section~\ref{case:cz-chordless}, define $t_1,\dots,t_x$, 
$G_R$, $G_B$, $G_Q$, $G_L$ and $G_0$.  See also Figure~\ref{fig:case3-splitting}.
Recall that $G_R$ satisfies all conditions with respect to corners $A_R:=A$, $B_R:=B$ and $C_R\in \{u_1,u_2\}$.
Apply induction on $G_R$ and obtain an
\int{(u_1, u_2)} representation $R_R$ with respect to the corners of $G_R$.  We use as layout for $R_R$ the type 
that we want for $G$, i.e., use a 3-sided/reverse 3-sided layout if we want $G$ to 
have a 3-sided/reverse 3-sided representation.

For $G_0$ and $G_Q$, we use exactly the same representations $R_0$ and $R_Q$
as in Section~\ref{case:cz-chordless}.

Combine now these three representations $R_R$, $R_Q$ and $R_0$ as described in Section~\ref{case:cz-chordless}, Case 3(a);
this can be done since the relevant curves $\bb{u_2^{R_R}},\dots,\bb{u_t^{R_R}}$
all end vertically in $R_R$.  
See also Figure~\ref{fig:case-3}.
The only change occurs at curve $\bb{C}$; in
Section~\ref{case:cz-chordless} this received a bend and a downward segment,
but here we omit this bend and segment and let $\bb{C}$ end horizontally
as desired.

One easily verifies that the curves intersect the bounding boxes as desired.
The constructed representations contain private regions for all interior faces of $G_R$,
$G_Q$ and $G_0$ by induction. The remaining faces are of the form $(C,u_i,u_{i+1}), 1 \leq i < q$,
and $(u_j, w_k, w_{k+1})$ where $w_{k}$ and $w_{k+1}$ are two consecutive neighbours of $u_j$
on the outer-face of $G_0$ or $G_Q$.
Private regions for those faces are shown in 
Figure~\ref{fig:case-3}.

\begin{figure}[ht]
\centering
\includegraphics[width=.99\textwidth, trim={0 13.8cm 0 0}, clip]{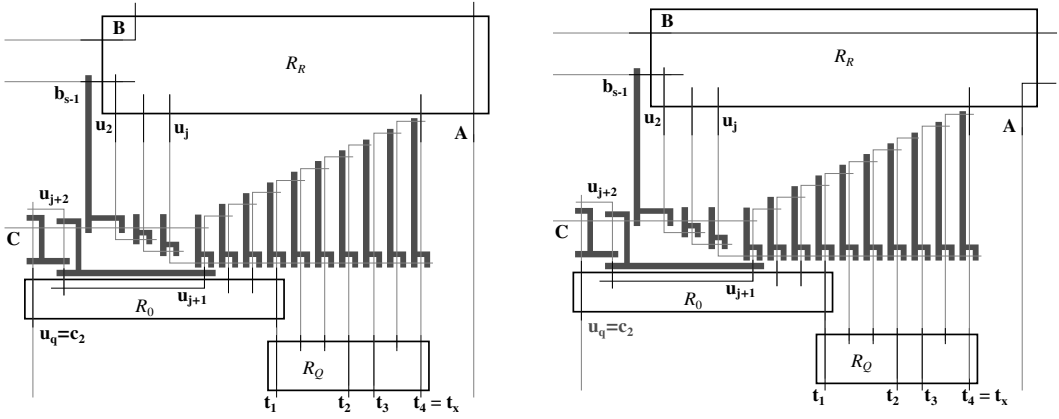}
\caption{Case 3(a): 3-sided representation when $deg(C) \geq 3$, there is
no chord incident to $C$, $F = \{(C,c_2)\}$, and $j > 1$.
The construction matches the graph depicted in Figure~\ref{fig:case3-splitting} left.
}
\label{fig:case-3}
\end{figure}


\medskip{\noindent{\bf {Case 3(b): $j=1$, i.e., there exists a chord $(b_{s-1},c_i)$. }}}
In this case we cannot use the above construction directly since we 
need to bend $\bb{u_j}=\bb{u_1}=\bb{b_{s-1}}$ horizontally rightwards to create 
intersections, but then it no longer extends vertically downwards as required
for $\bb{b_{s-1}}$.   
The simple construction described in Section~\ref{case:cz-chordless}, Case 3(b) does not apply either. However, if we use a
different vertex as $u_j$ (and argue carefully that the chord condition holds), then
the same construction works.

\begin{figure}[ht]
\centering
\includegraphics[width=.49\textwidth]{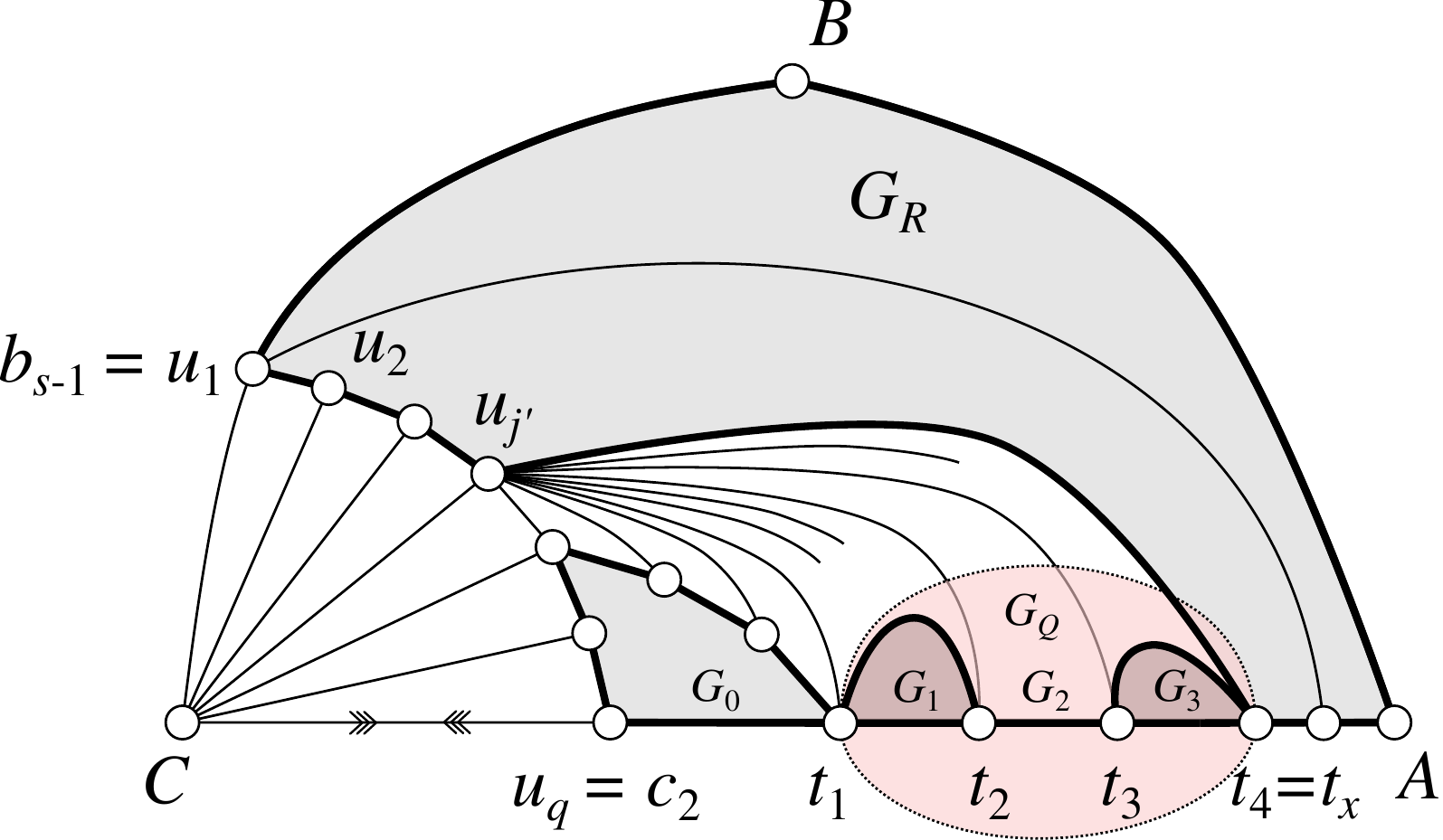}
\caption{Case 3(b): Splitting the graph when $deg(C) \geq 3$, no chord is 
incident to $C$, and $j=1$.
} 
\label{fig:case3-splitting-j1}
\end{figure}

Recall that $u_1, \ldots, u_q$ are the
neighbours of corner $C$ in CW order starting with $b_{s-1}$ and ending with $c_2$. 
We know that $q \geq 3$ and $u_2,\dots,u_{q-1}$ are not on the outer-face. 
Now define $j'$ as follows: Let $u_{j'}$, ${j'} > 1$ be a neighbour
of $C$ that has at least one neighbour on $P_{CA}$ other than $C$, and choose $u_{j'}$
so that ${j'}$ is minimal while satisfying ${j'}>1$.  
Such a $j'$ exists since $u_{q-1}$ 
has another neighbour on $P_{CA}$, and by $q\geq 3$ we have $q-1>1$. 
Now, separate $G$ as in the previous case, except use
$j'$ in place of $j$.  Thus, define $t_1,\dots,t_x$ to be the neighbours of $u_{j'}$ on $P_{c_2A}$, in order,
and separate $G$ into three graphs as follows:

\begin{itemize}
\item The \emph{right} graph $G_R$ 
is bounded by $(A, \stackrel{P_{AB}}{\ldots}, B, \stackrel{P_{Bu_1}}{\ldots}, u_1,u_2,\dots,u_{j'}, t_x, \stackrel{P_{t_xA}}{\ldots},A)$.
\item Let $G_B$ be the graph bounded by $(u_{j'}, t_1, \stackrel{P_{t_1t_x}}{\ldots}, t_x,u_{j'})$. 
	Define $G_Q:= G_B-u_{j'}$.
\item Let $G_L$ be the graph bounded by $(C, \stackrel{P_{Ct_1}}{\ldots},t_1,u_{j'},C)$.
	Define $G_0:=G_L-\{u_{j'},C\}$.
\end{itemize}

Observe that the boundaries of all the graphs are simple cycles, and thus they are
W-triangulations. Select $(A_R := A, B_R := B, C_R := u_2)$ to be the corners of $G_R$
and argue the chord condition as follows:

\begin{itemize}
    \item $G_R$ does not have any chords on $P_{C_RA_R}$ as such chords
   would either contradict the minimality of $j'$, or violate the chord condition
   in $G$.
    \item $G_R$ does not have any chords on $P_{A_RB_R} = P_{AB}$.
    \item $G_R$ does not have any chords on $P_{Bb_{s-1}}$ as it is a sub-path
    of $P_{BC}$ and they would violate the chord condition in $G$. It also does not
    have any chords in the form $(C_R = u_2, b_\ell), 1 \leq \ell < s-1$ as they would 
    have to intersect the chord $(b_{s-1},c_i)$, violating the planarity of $G$. Hence,
    $G_R$ does not have any chords on $P_{C_RA_R}$. 
	\item
	Notice in particular that the chord $(u_1,c_i)$ of $G_R$ is {\em not} a violation of
	the chord condition since we chose $u_2$ as a corner.
\end{itemize}

Hence, we can obtain a representation $R_R$ of $G_R$ with 3-sided or reverse 3-sided layout
and special edge $(u_1 = b_{s-1}, u_2)$. 
For  graphs $G_Q$ and $G_0$ the corners are chosen, 
the chord condition is verified, and the representations are obtained exactly as in Case 3(a).
Since the special edge of $G_R$ is $(u_1,u_2)$ as before, curves
$\bb{u_1}$ and $\bb{u_2}$ are situated precisely as in Case 3(a),
and we merge representations and find private regions as before.

\bigskip
This ends the description of the construction in all cases, and
hence proves Lemma~\ref{lem:3-sided}.

\section{From 4-Connected Triangulations to All Planar Graphs}

In this section, we prove Theorem~\ref{thm:main-claim}.  Observe that
Lemma~\ref{lem:3-sided} essentially proves it for 4-connected triangulations.
As in \cite{cit:chalopin-string} we extend it to all triangulations by induction
on the number
of separating triangles.

\begin{figure}[ht]
\centering
\includegraphics[width=.20\textwidth]{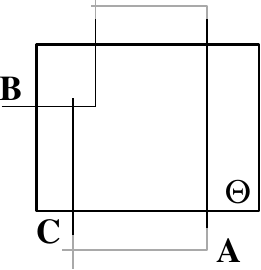}
\caption{Completing a 3-sided \int{(B,C)} representation 
by adding intersections for $(A,B)$ and $(A,C)$.}
\label{fig:completion}
\end{figure}

\begin{theorem}
\label{thm:triangulations}
Let $G$ be a triangulation with outer-face $(A,B,C)$. $G$ has a $1$-string $B_2$-VPG representation
with a chair-shaped private region for every interior face $f$ of $G$.
\end{theorem}
\begin{proof}
Our approach is exactly the same as in \cite{cit:chalopin-string}, except that
we must be careful not to add too many bends when merging subgraphs at separating triangles,
and hence must use 3-sided layouts.  Formally,
we proceed by induction on the number of separating triangles. In the base case, $G$ has
no separating triangle, i.e., it is 4-connected. As the outer-face is a triangle, $G$ clearly satisfies the chord 
condition. Thus, by Lemma~\ref{lem:3-sided}, it has a 3-sided \int{(B,C)} representation $R$
with private region for every face.  $R$ has an intersection for 
every edge except for $(A,B)$ and $(A,C)$. 
These intersections can be created by tangling $\bb{B},\bb{A}$ and
$\bb{C},\bb{A}$ suitably (see Figure~\ref{fig:completion}).
Recall that $\bb{A}$ initially did not have any bends, so it has 2 bends in the
constructed representation of $G$. The existence of private regions is guaranteed
by Lemma~\ref{lem:3-sided}.

Now assume for induction that $G$ has $k + 1$
separating triangles. Let $\Delta = (a,b,c)$ be an inclusion-wise minimal separating triangle of $G$.
It was shown in~\cite{cit:chalopin-string} that the subgraph $G_2$ induced by the vertices inside $\Delta$ 
is either an isolated vertex, or a W-triangulation with corners $(A,B,C)$ such that
the vertices on $P_{AB}$ are adjacent to $b$, the vertices on $P_{BC}$ are adjacent to $c$, and 
the vertices on $P_{CA}$ are adjacent to $a$. Furthermore, $G_2$ satisfies the chord
condition. Also, graph $G_1 = G - G_2$ is a W-triangulation that satisfies the chord
condition and has $k$ separating triangles. By induction, $G_1$ has a representation $R_1$ (with respect to the corners of $G_1$)
with a chair-shaped private region for every interior face $f$. Let $\Phi$ be the private region for face $\Delta$. 
Permute $a,b,c$, if needed, so that the naming corresponds to the one needed for the private region and, in particular, the vertical segment of $\bb{c}$ intersects the private region of $\Delta$ as depicted in Figure~\ref{fig:separating-triangle}.

{\medskip{\noindent{\bf {Case 1: ${G_2}$ is a single vertex ${v}$. }}}}
 Represent $v$ by inserting into $\Phi$ an orthogonal curve $\bb{v}$ with 2 bends 
that intersects $\bb{a},\bb{b}$ and $\bb{c}$.
The construction, together with private regions for the newly created
faces $(a,b,v)$, $(a,c,v)$ and $(b,c,v)$,
is shown in Figure~\ref{fig:separating-triangle}.

\medskip{\noindent{\bf {Case 2: ${G_2}$ is a W-triangulation. }}}
Recall that $G_2$ satisfies the chord condition with respect to corners $(A,B,C)$.
Apply Lemma~\ref{lem:3-sided} to construct a 3-sided \int{(C,b_{s-1})} representation $R_2$ of $G_2$ with respect to the corners of $G_2$. 
Let us assume that (after possible rotation) $\Phi$ has the orientation shown
in Figure~\ref{fig:separating-triangle} (right); if it had the symmetric orientation
then we would do a similar construction using a reverse 3-sided representation of $G_2$.
Place $R_2$ inside $\Phi$ as shown in Figure~\ref{fig:separating-triangle} (right).
Stretch the curves representing vertices on $P_{CA}$, $P_{AB}$ and $P_{Bb_{s-1}}$
downwards, upwards and leftwards respectively so that they intersect $\bb{a}, \bb{b}$ and $\bb{c}$.
Top-tangle leftwards the curves $\bb{A} = \bb{a_1}, \bb{a_{2}}, \ldots, \bb{a_r} = \bb{B}$.
Left-tangle downwards the curves $\bb{B} = \bb{b_1}, \bb{b_2}, \ldots, \bb{b_{s-1}}$
and bend and stretch $\bb{C}$ downwards so that it intersects $\bb{a}$.
Bottom-tangle leftwards the curves $\bb{C} = \bb{c_1}, \ldots, \bb{c_t} = \bb{A}$.
It is easy to verify that the construction creates intersections for all the edges between vertices
of $\Delta$ and the outer-face of $G_2$. The tangling operation then creates intersections for all the 
outer-face edges of $G_2$ except edge $(C,b_{s-1})$, which is already represented in $R_2$.

Every curve that receives a new bend represents a vertex on the outer-face of $G_2$, which
means that it initially had at most 1 bend. Curve $\bb{A}$ is the only curve that
receives 2 new bends, but this is allowed as $\bb{A}$ does not have any bends in $R_2$. 
Hence, the number of bends for every curve does not exceed 2. 

Private regions for faces formed by vertices $a,b,c$ and vertices
on the outer-face of $G_2$ can be found as shown in Figure~\ref{fig:separating-triangle} right.
\end{proof}

\begin{figure}[ht]
\centering
\includegraphics[width=.30\textwidth]{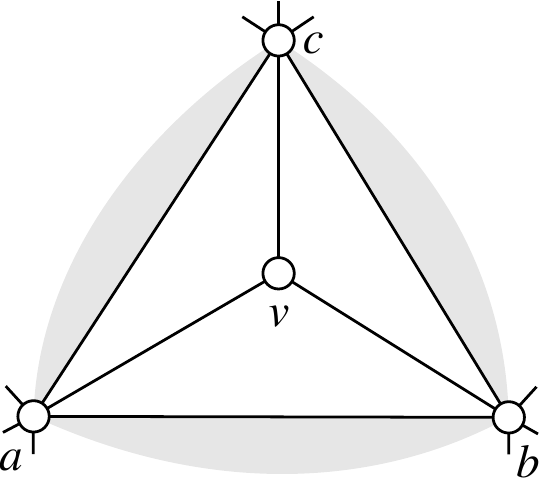}\hspace{1em}
\includegraphics[width=.45\textwidth]{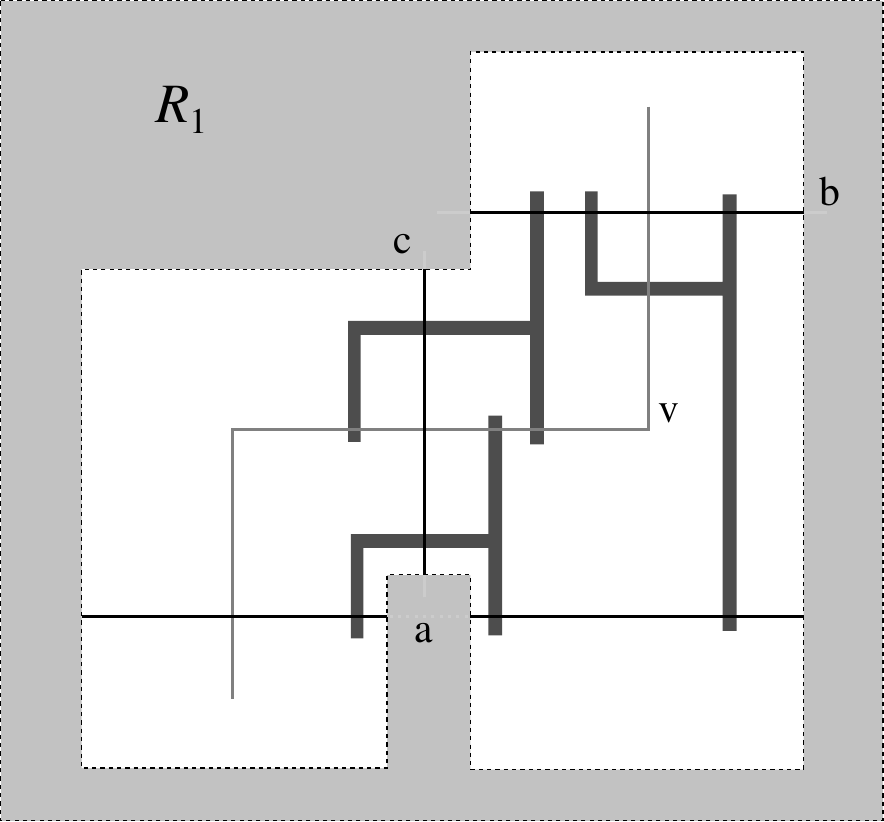}\\\medskip
\includegraphics[width=.30\textwidth]{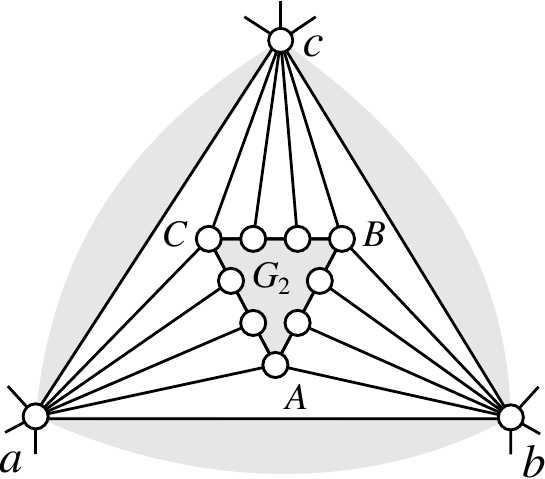}\hspace{1em}
\includegraphics[width=.45\textwidth]{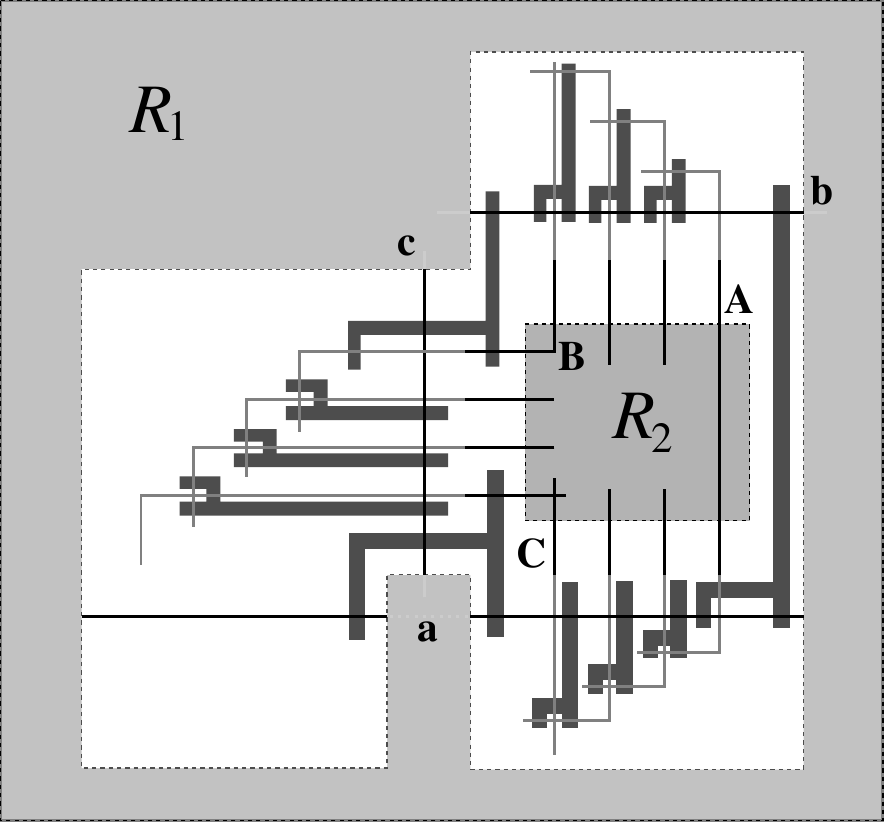}
\\\vspace{1em}
\caption{A separating triangle enclosing one vertex and the construction (top), and a
separating triangle enclosing a W-triangulation and the corresponding construction (bottom).}
\label{fig:separating-triangle}
\end{figure}

With Theorem~\ref{thm:triangulations} in hand, we can show our main result: every planar graph 
has a 1-string $B_2$-VPG representation.

\begin{proof}[Proof of Theorem~\ref{thm:main-claim}]
If $G$ is a planar triangulated graph, then the claim
holds by Theorem~\ref{thm:triangulations}.
To handle an arbitrary planar graph, repeatedly stellate the graph
(recall that this means inserting into each non-triangular face a
new vertex connected to all vertices of the face).
It is easily shown that one stellation makes the graph connected,
a second one makes it 2-connected, and a third one makes it
3-connected and triangulated.  Thus after 3 stellations we have 
a 3-connected triangulated graph $G'$ such that $G$ is an induced subgraph of $G'$. 
Apply Theorem~\ref{thm:triangulations} to construct a $1$-string $B_2$-VPG representation $R'$ of $G'$ (with the three outer-face vertices chosen as corners).
By removing curves representing vertices that are not in $G$, we obtain a $1$-string $B_2$-VPG representation of $G$.
\end{proof}

\section{Conclusions and Outlook}
\label{sec:conclusions}
\label{sec:outlook}

We showed that every planar graph has a 1-string $B_2$-VPG representation, i.e.,
a representation as an intersection graph of strings where strings cross at most once
and each string is orthogonal with at most two bends.    One advantage of this
is that the coordinates to describe such a representation are small,  since orthogonal
drawings can be deformed easily such that all bends are at integer coordinates.
Every vertex curve has at most two bends and hence at most 3 segments, so
the representation can be
made to have coordinates in an $O(n)\times O(n)$-grid with perimeter at most $3n$.
Note that none of the previous results provided an intuition of the required size of the grid.

Following the steps of our proof, it is not hard to see that our representation
can be found in linear time, since the only non-local operation is to test whether
a vertex has a neighbour on the outer-face. This can be tested by marking such
neighbours whenever they become part of the outer-face.  Since no vertex ever is
removed from the outer-face (updating the outer-face markers upon removing such vertices could increase the time complexity), this takes overall linear time.

The representation
constructed in this paper uses curves of 8 possible shapes for planar graphs.  
For 4-connected planar graphs, the shapes that have at most one vertical
segment suffice.
A natural question is if one can restrict the number of shapes
required to represent all planar graphs.

Bringing this effort further, is it possible to restrict the curves even more?  The existence of $1$-string $B_1$-VPG representations for planar graphs is open. Furthermore, Felsner et al.~\cite{cit:mfcs} asked the question whether every planar
graph is the intersection graph of only two shapes, namely $\{L,\Gamma\}$. As they point out,
a positive result would provide a different proof of Scheinerman's conjecture (see~\cite{cit:stretch} for details).
Somewhat inbetween: is every planar graph the intersection graph of
$xy$-monotone orthogonal curves, preferably in the 1-string model
and with few bends?


\bibliography{socg}

\appendix

\section{Example}

Here we provide an example of constructing an \int{(18,16)} $1$-string $B_2$-VPG representation $R$ of the W-triangulation shown in Figure~\ref{fig:example-main}.   We use numbers and colors to distinguish vertices.  We use letters to 
indicate special vertices such as corners; note that the designation as such
a corner may change as the subgraph gets divided further.
The special edge is marked with hatches.

One can verify that the graph with the chosen corners (1,4,18) satisfies the chord condition. Vertex $C$ has degree 3, but it is not incident to a chord, so one applies the construction from Section~\ref{case:chordless}. Finding vertex $u_j$, we can see that $j > 1$, so Case 3(a) applies. Figure~\ref{fig:example-main} shows the graphs $G_R$, $G_Q$ and $G_0$, and how to construct $R$ from their representations $R_R$, $R_Q$ and $R_0$. 

\begin{figure}[ht]
\includegraphics[width=\textwidth]{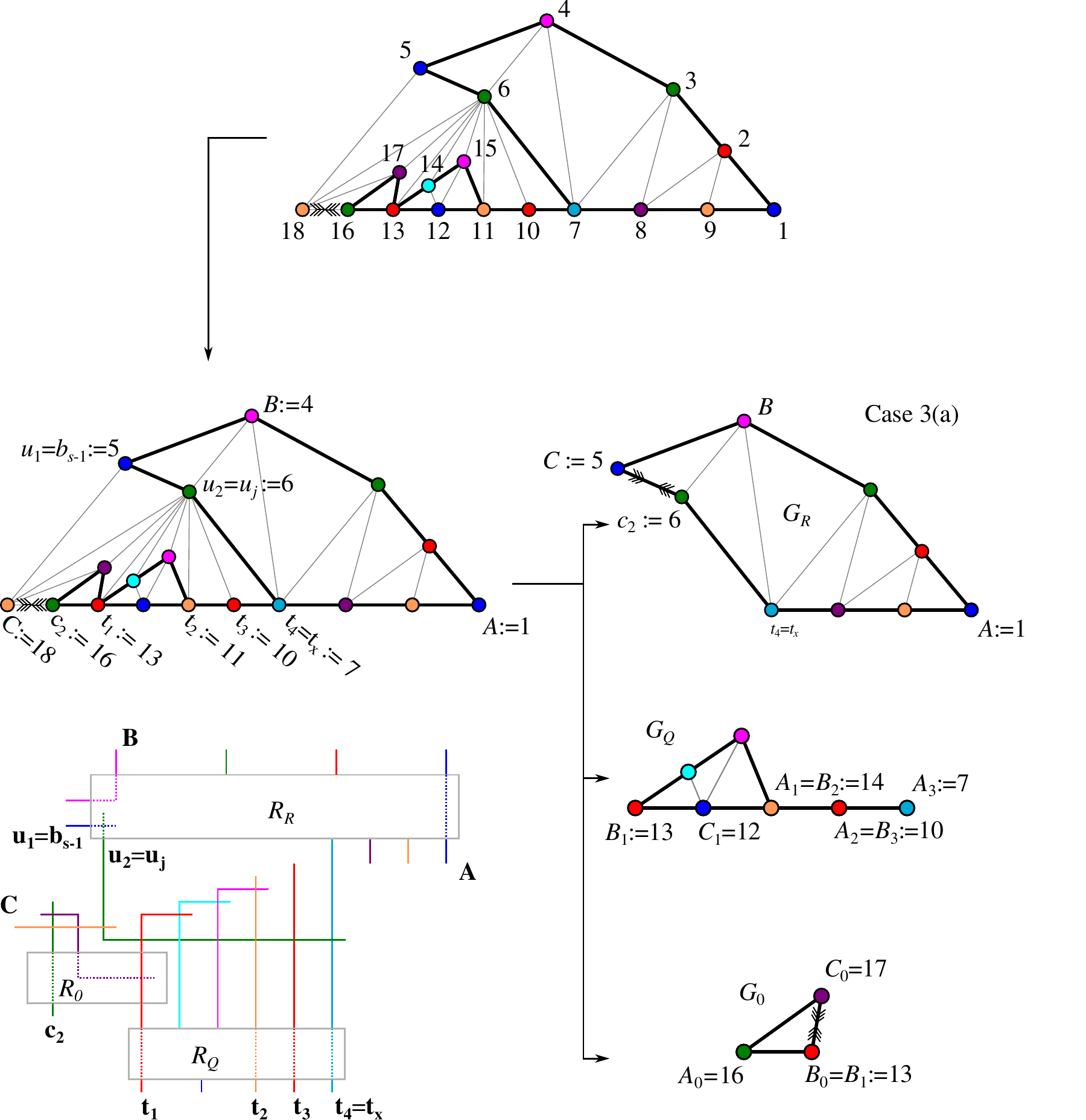}
\caption{Illustration of the example. The goal is to find an \int{(18,16)} $1$-string $B_2$-VPG representation of the W-triangulation shown on top, using corners (1,4,18).}
\label{fig:example-main}
\end{figure}

The construction of $R_Q$ is shown in Figure~\ref{fig:example-path}. The representation should have a $2$-sided layout and no special edge. Graph $G_Q$ decomposes into three subgraphs $G_1, G_2, G_3$. Their 2-sided representations are found separately and combined as described in the proof of Claim~\ref{claim}.

\begin{figure}[ht]
\includegraphics[width=\textwidth]{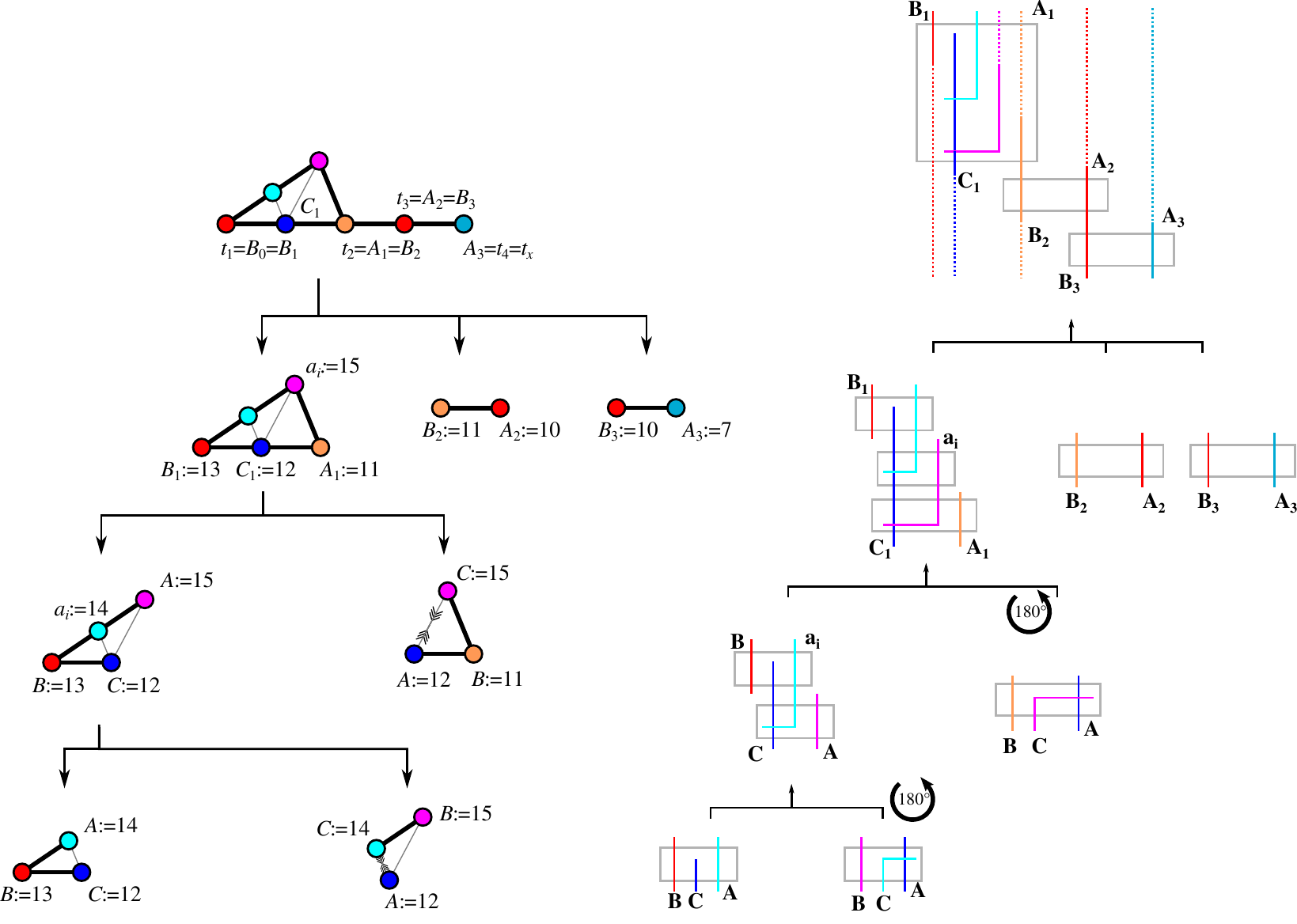}
\caption{Illustration of the example: Finding $R_Q$ (top right).}
\label{fig:example-path}
\end{figure}

The construction of $R_R$ is shown in Figure~\ref{fig:gr-decompose} (decomposition of $G_R$) and~\ref{fig:gr-compose} (combining the representations). Representation $R_R$ is supposed to be $3$-sided. We first apply Case 1 (Section~\ref{case:c-degree-2-2-sided}) twice, since corner $C$ has degree 2. 
Then corner $C$ becomes incident to a chord, so we are in
Case 2, and use sub-case Case 2(a) (Section~\ref{case:special})
since the special edge is $(C, b_{s-1}=B)$.
This case calls for a 3-sided representation of a $G_2$ (which is a triangle in this case, so the base case applies).  It also calls for a $2$-sided representation of $G_1$ with special edge $(C,A=c_2)$. 
This is Case~2 (Section~\ref{case:cz-special}) and we need to apply the reversal trick---we flip the graph and relabel the corners. After obtaining the representation, it must be flipped horizontally in order to undo the reversal. The construction decomposes the graph further, using Case~2 repeatedly, which breaks the graphs into elementary triangles. Their 2-sided representations are obtained using the base case and composed as stipulated by the construction.

\begin{figure}[ht]
\includegraphics[width=\textwidth]{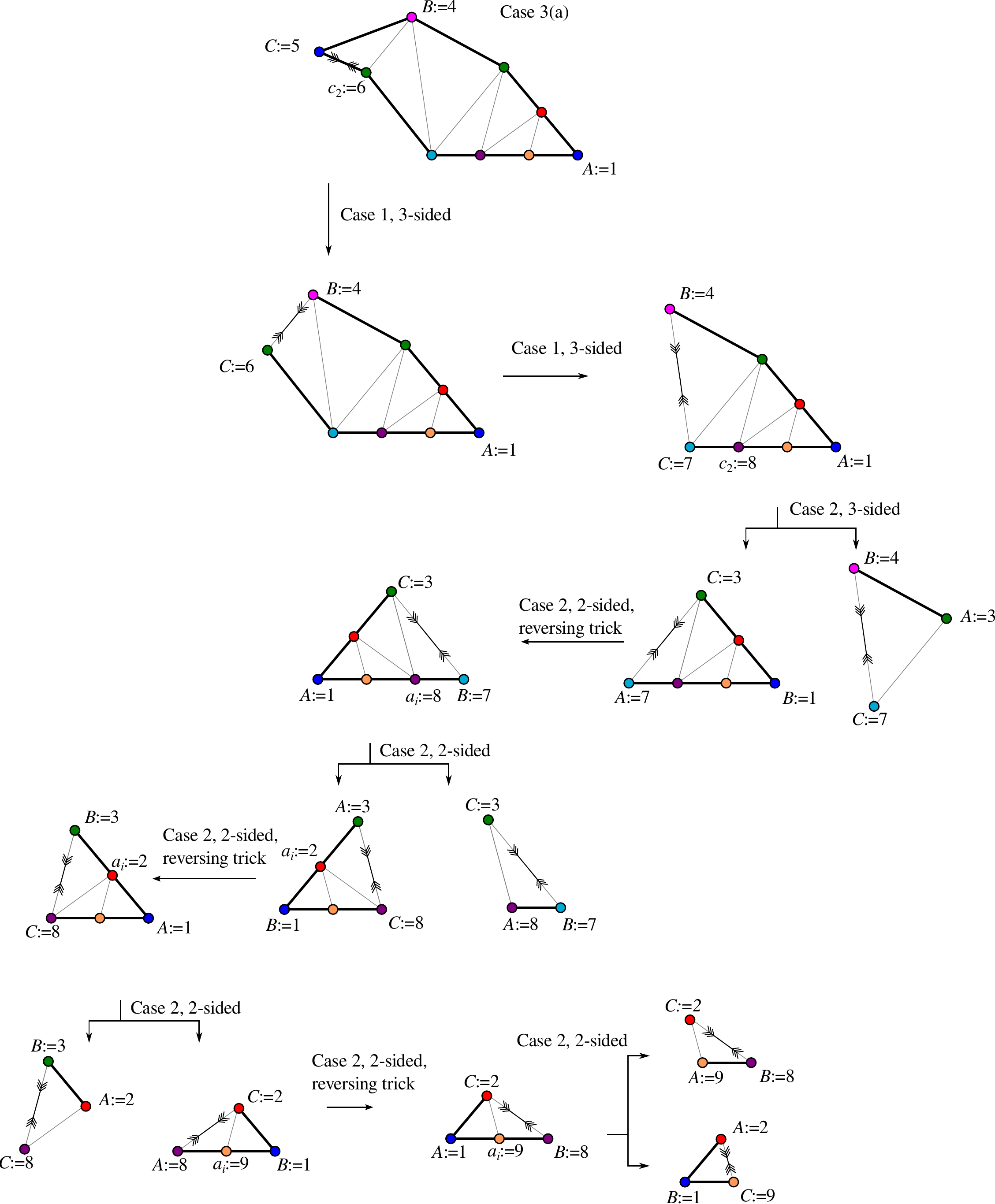}
\caption{Illustration of the example: Decomposing graph $G_R$.}
\label{fig:gr-decompose}
\end{figure}

\begin{figure}[ht]
\includegraphics[width=\textwidth]{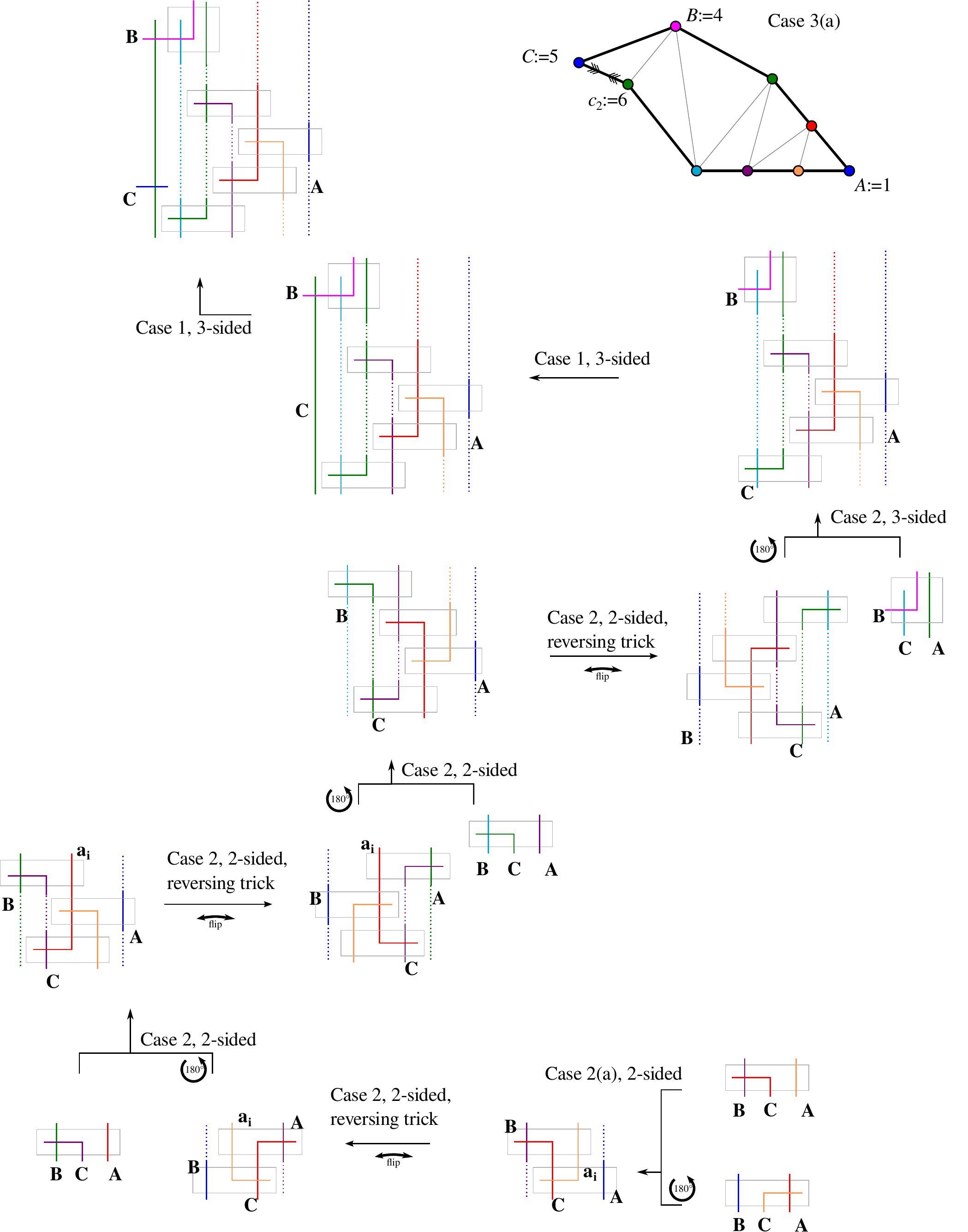}
\caption{Illustration of the example: Composing representation $R_R$.}
\label{fig:gr-compose}
\end{figure}


Figure~\ref{fig:complete} shows the complete 3-sided \int{(18,16)} representation of the graph.

\begin{figure}[ht]
\centering
\includegraphics[width=.75\textwidth]{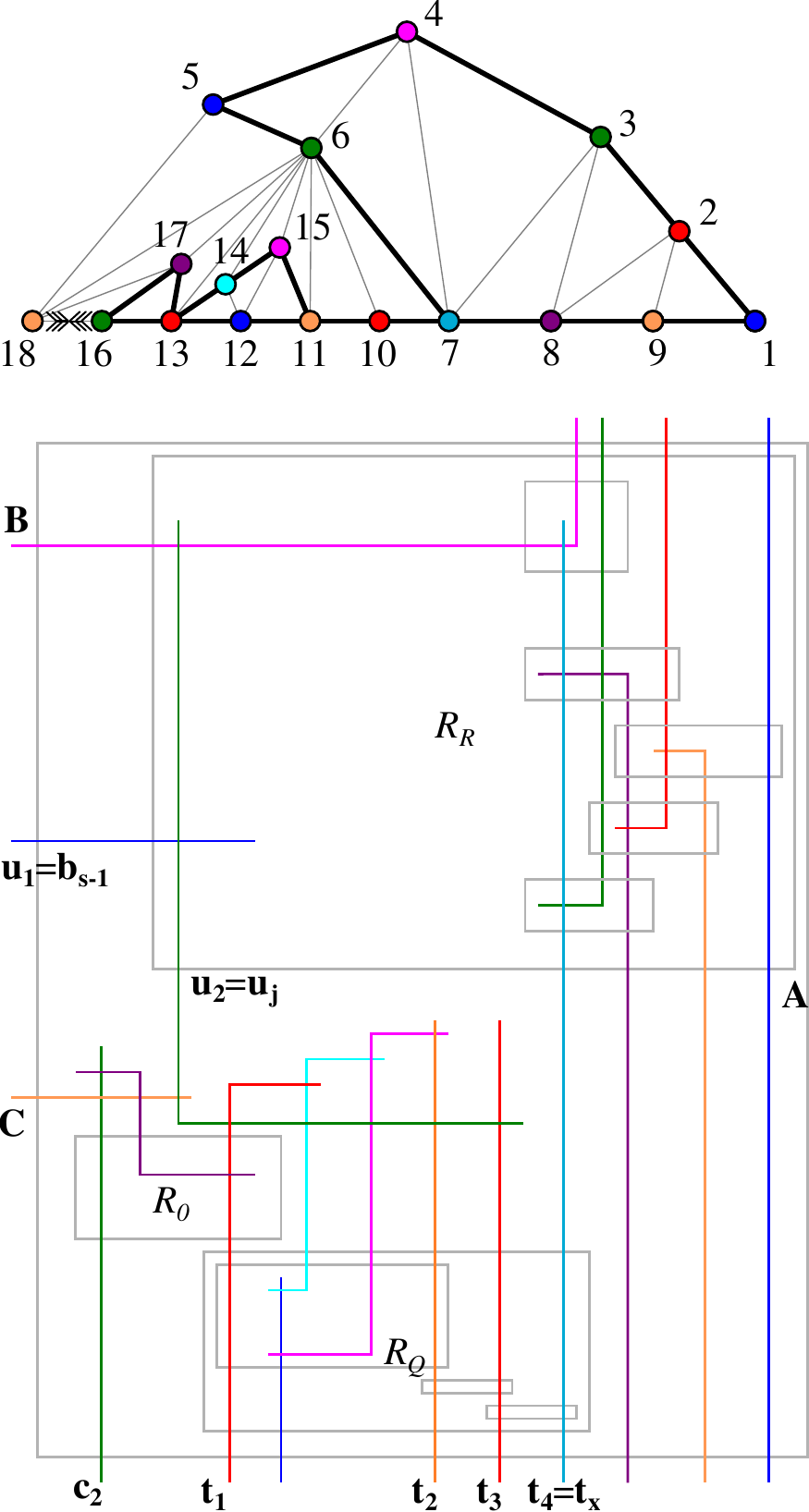}
\caption{Illustration of the example: Complete 3-sided \int{(18,16)} representation.}
\label{fig:complete}
\end{figure}

\end{document}